\documentclass{svproc}
\usepackage{multirow}
\usepackage{array}
\usepackage{wrapfig}

\usepackage{amsmath}
\usepackage{amssymb}
\usepackage{xspace}
\usepackage{enumitem}
\usepackage{IEEEtrantools}

\DeclareMathOperator{\bL}{\mathbf{K}}

\DeclareMathOperator{\bK}{\mathbf{K}}
\DeclareMathOperator{\bM}{\mathbf{M}}

\newcommand\sfequilibrium{\mbox{S5-equilibrium}\xspace}

\newcommand\kdequilibrium{\mbox{equilibrium}\xspace}

\newcommand\sfinterpretation{\mbox{S5-interpretation}\xspace}

\newcommand\sfmodel{\mbox{S5-model}\xspace}

\newcommand\kdmodel{\mbox{D45-model}\xspace}

\newcommand\sfhtinterpretation{\mbox{\text{S5}-interpretation}\xspace}
\newcommand\sfhtinterpretations{\mbox{\text{S5}-interpretations}\xspace}

\newcommand\wv{\mathbb{W}}
\newcommand\us{\mathbb{S}}
\newcommand\sfhtmodel{\mbox{$\text{S5}_\text{\tiny HT}$-model}\xspace}

\newcommand\kdint[2]{({#2, #1})}

\def\MEQ{{\rm {EQB}}}
\def\WEQ{{\rm {WEQ}}}
\def\EQ{{\rm {EQ}}}

\def\Head{\mathit{Head}}
\def\Body{\mathit{Body}}
\def\Bodym{\mathit{Body}^{sub}}
\def\Bodyr{\mathit{Body}^{obj}}
\def\Bodyp{\mathit{Body}^{+}}

\def\Bodyrp{\mathit{Body}^+_{ob}}
\def\Bodymp{\mathit{Body}^+_{sub}}

\def\modelsm{\models_{\text{\rm \tiny m}}}
\def\modelshtm{\models_{\text{\rm \tiny htm}}}

\def\modelsm{\models_{\text{\rm \tiny m}}}

\def\SM{\text{\rm SM}}
\def\CL{\text{\rm CL}}

\newcolumntype{L}[1]{>{\raggedright\let\newline\\\arraybackslash\hspace{0pt}}m{#1}}
\newcolumntype{C}[1]{>{\centering\let\newline\\\arraybackslash\hspace{0pt}}m{#1}}
\newcolumntype{R}[1]{>{\raggedleft\let\newline\\\arraybackslash\hspace{0pt}}m{#1}}

\usepackage[comments=on,changes=on]{krudces}
\title{Founded World Views with\\ Autoepistemic Equilibrium Logic}
\author{
Pedro Cabalar \inst{1} \and
Jorge Fandinno \inst{2} \and 
Luis Fari{\~n}as \inst{2}
}
\institute{University of Corunna, Spain\\
\email{cabalar@udc.es}
\and University of Toulouse IRIT, CNRS, France\\
\email{\{jorge.fandinno,farinas\}@irit.fr}
}

\begin{document}

\maketitle

\begin{abstract}
\noindent
Defined by Gelfond in 1991 (G91), epistemic specifications (or programs) are an extension of logic programming under stable models semantics that introduces \emph{subjective literals}.
A subjective literal allows checking whether some regular literal is true in all (or in some of) the stable models of the program, being those models collected in a set called \emph{world view}.
One epistemic program may yield several world views but, under the original G91 semantics, some of them resulted from self-supported derivations.
During the last eight years, several alternative approaches have been proposed to get rid of these self-supported world views.
Unfortunately, their success could only be measured by studying their behaviour on a set of common examples in the literature, since no formal property of ``self-supportedness'' had been defined.
To fill this gap, we extend in this paper the idea of unfounded set from standard logic programming to the epistemic case.
We define when a world view is \emph{founded} with respect to some program and propose the \emph{foundedness} property for any semantics whose world views are always founded.
Using counterexamples, we explain that the previous approaches violate foundedness, and proceed to propose a new semantics based on a combination of Moore's Autoepistemic Logic and Pearce's Equilibrium Logic.
The main result proves that this new semantics precisely captures the set of founded G91 world views.
\end{abstract}

\section{Introduction}\label{sec:introduction}

The language of \emph{epistemic specifications}, proposed by Gelfond in 1991~\cite{Gelfond91}, extends disjunctive logic programs (under the \emph{stable model}~\cite{GL88} semantics) with modal constructs called \emph{subjective literals}.
Using these constructs, it is possible to check whether a regular literal $l$ is true in \emph{every} stable model (written $\bK l$) or in \emph{some} stable model (written $\bM l$) of the program.
For instance, the rule:
\begin{eqnarray}
a \leftarrow \neg\!\bK b \label{f:loop1}
\end{eqnarray}
means that $a$ must hold if we cannot prove that all the stable models contain~$b$.
Subjective literals have been incorporated as an extension of the Answer Set Programming (ASP) paradigm~\cite{MT99,Nie99} in different solvers and implementations -- see~\cite{LK18} for a recent survey.
The definition of a ``satisfactory'' semantics for epistemic specifications has proved to be a non-trivial enterprise, as shown by the list of different attempts proposed so far~\cite{CerroHS15,Gelfond91,Gelfond11,Kahl15,ShenE17,Truszczynski11,WangZ05}.
The main difficulty arises because subjective literals query the set of stable models but, at the same time, occur in rules that determine those stable models.
As an example, the program consisting of:
\begin{eqnarray}
b \leftarrow \neg\!\textbf{}\bK a \label{f:loop2}
\end{eqnarray}
and \eqref{f:loop1} has now two rules defining atoms $a$ and $b$ in terms of the presence of those same atoms in all the stable models.
To solve this kind of cyclic interdependence, the original semantics by Gelfond~\cite{Gelfond91} (G91) considered different alternative \emph{world views} or sets of stable models.
In the case of program \eqref{f:loop1}-\eqref{f:loop2}, G91 yields two alternative world views\footnote{For the sake of readability, sets of propositional interpretations are embraced with $\sset{ \ }$ rather than $\set{\ }$.}, $[\{a\}]$ and $[\{b\}]$, each one containing a single stable model, and this is also the behaviour obtained in the remaining approaches developed later on.
The feature that made G91 unconvincing, though, was the generation of self-supported world views.
A prototypical example for this effect is the epistemic program consisting of the single rule:
\begin{eqnarray}
a \leftarrow \bK a \label{f:self}
\end{eqnarray}
whose world views under G91 are $[\emptyset]$ and $[\{a\}]$.
The latter is considered as counter-intuitive by all authors\footnote{This includes Gelfond himself, who proposed a new variant in~\cite{Gelfond11} motivated by this same example and further modified this variant later on in~\cite{Kahl15}.} because it relies on a self-supported derivation: $a$ is derived from $\bK a$ by rule \eqref{f:self}, but the only way to obtain $\bK a$ is rule \eqref{f:self} itself.
Although the rejection of world views of this kind seems natural, the truth is that all approaches in the literature have concentrated on studying the effects on individual examples, rather than capturing the absence of self-supportedness as a formal property.
To achieve such a goal, we would need to establish some kind of \emph{derivability} condition in a very similar fashion as done with \emph{unfounded sets}~\cite{GelderRS91} for standard logic programs.
To understand the similarity, think about the (tautological) rule $a \leftarrow a$.
The classical models of this rule are $\emptyset$ and $\{a\}$, but the latter cannot be a stable model because $a$ is not derivable applying the rule.
Intuitively, an unfounded set is a collection of atoms that is not derivable from a given program and a fixed set of assumptions, as happens to $\{a\}$ in the last example.
As proved by~\cite{LeoneRS97}, the stable models of any disjuntive logic program are precisely its classical models that are \emph{founded}, that is, that do not admit any unfounded set.
As we can see, the situation in~\eqref{f:self} is pretty similar to $a \leftarrow a$ but, this time, involves derivability through subjective literals.
An immediate option is, therefore, extending the definition of unfounded sets for the case of epistemic programs -- this constitutes, indeed, the first contribution of this paper.

Once the property of \emph{founded} world views is explicitly stated, the paper proposes a new semantics for epistemic specifications, called \emph{Founded Autoepistemic Equilibrium Logic} (FAEEL), that fulfills that requirement.
In the spirit of~\mbox{\cite{CerroHS15,WangZ05}}, our proposal actually constitutes a full modal non-monotonic logic where $\bK$ becomes the usual necessity operator applicable to arbitrary formulas.
Formally, FAEEL is a combination of Pearce's \emph{Equilibrium Logic}~\cite{Pearce96}, a well-known logical characterisation of stable models, with Moore's \emph{Autoepistemic Logic} (AEL)~\cite{Moore85}, one of the most representative approaches among modal non-monotonic logics.
The reason for choosing Equilibrium Logic is quite obvious, as it has proved its utility for characterising other extensions of ASP, including the already mentioned epistemic approaches in~\mbox{\cite{CerroHS15,WangZ05}}.
As for the choice of AEL, it shares with epistemic specifications the common idea of \emph{agent's introspection} where $\bL \varphi$ means that $\varphi$ is one of the agent's beliefs.
The only difference is that those beliefs are just classical models in the case of AEL whereas epistemic specifications deal with stable models instead.
Interestingly, the problem of self-supported models has also been extensively studied in AEL~\cite{Konolige88,MarekT89,Niemela91,Schwarz91}, where the formula $\bL a \to a$, analogous to \eqref{f:self}, also yields an unfounded world view\footnote{Technically, AEL is defined in terms of \emph{theory expansions} but each one can be characterised by a canonical $\sfmodel$ with the same form of a world view~\cite{Moore84,schwarz1992minimal}.
}
$[\{a\}]$.
Our solution consists in combining the monotonic bases of AEL and Equilibrium Logic (the modal logic $\text{KD45}$ and the intermediate logic of Here-and-There (HT)~\cite{Hey30}, respectively), but defining a two-step models selection criterion that simultaneously keeps the agent's beliefs as stable models and avoids unfounded world views from the use of the modal operator $\bL$.
As expected, we prove that FAEEL guarantees the property of founded world views, among other features.
Our main result, however, goes further and asserts that the FAEEL world views of an epistemic program are precisely the set of founded G91 world views.
We reach, in this way, an analogous situation to the case of standard logic programming, where stable models are the set of founded classical models of the program.

The rest of the paper is organised as follows.
Section~\ref{sec:background} and~\ref{sec:g91} respectively revisit the background knowledge about equilibrium logic and epistemic specifications necessary for the rest of the paper.
Section~\ref{sec:unfounded} introduces the foundedness property for epistemic logic programs.
In Section~\ref{sec:autoepistemic.equilibrium.logic},
we introduce FAEEL and show that its world views precisely coincide with the set of founded G91 world views.
Finally, Section~\ref{sec:conclusions} concludes the paper.

\section{Background}
\label{sec:background}

We begin recalling the basic definitions of equilibrium logic and its relation to stable models.
We start from the syntax of propositional logic, with formulas built from combinations of atoms in a set $\at$ with operators $\wedge, \vee, \bot$ and $\to$ in the usual way.
We define the derived operators
$\fF \leftrightarrow \fG \eqdef (\fF \to \fG) \wedge (\fG \to \fF)$,
\ $(\fF \leftarrow \fG) \eqdef (\fG \to \fF)$,
\ $\neg \fF \eqdef (\fF \to \bot)$
and
\ $\top \eqdef \neg \bot$.
%

A \emph{propositional interpretation} $T$ is a set of atoms $T \subseteq \at$. 
%
We write $T \models \varphi$ to represent that $T$ classically satisfies formula $\varphi$. 
An \emph{\htinterpretation} is a pair $\tuple{H,T}$ (respectively called ``here'' and ``there'') of propositional interpretations such that $H \subseteq T \subseteq \at$; it is said to be \emph{total} when $H=T$.
We write $\tuple{H,T} \models \varphi$ to represent that $\tuple{H,T}$ \emph{satisfies} a formula~$\varphi$ under the recursive conditions:
\begin{itemize}[ topsep=2pt]
\item $\tuple{H,T} \not\models \bot$ 
\item $\tuple{H,T} \models p$ iff $p \in H$ 
\item $\tuple{H,T} \models \varphi \wedge \psi$ iff $\tuple{H,T} \models \varphi$ and $\tuple{H,T} \models \psi$
\item $\tuple{H,T} \models \varphi \vee \psi$ iff $\tuple{H,T} \models \varphi$ or $\tuple{H,T} \models \psi$
\item $\tuple{H,T} \models \varphi \to \psi$ iff both (i) $T \models \varphi \to \psi$ and (ii) $\tuple{H,T} \not\models \varphi$ or $\tuple{H,T} \models \psi$
\end{itemize}
As usual, we say that $\tuple{H,T}$ is a \emph{model} of a theory~$\Gamma$,
in symbols $\tuple{H,T} \models \Gamma$, iff $\tuple{H,T} \models \varphi$
for all $\varphi \in \Gamma$.
It is easy to see that $\tuple{T,T} \models \Gamma$ iff $T \models \Gamma$ classically.
For this reason, we will identify $\tuple{T,T}$ simply as $T$ and will use `$\models$' indistinctly.
By $\CL[\Gamma]$ we denote the set of all classical models of $\Gamma$.
Interpretation $\tuple{T,T}=T$ is a \emph{stable (or equilibrium) model} of a theory $\Gamma$ iff $T \models \Gamma$ and there is no $H\subset T$ such that $\tuple{H,T} \models \Gamma$.
We write $\SM[\Gamma]$ to stand for the set of all stable models of $\Gamma$.
Note that $\SM[\Gamma] \subseteq \CL[\Gamma]$ by definition.

\section{G91 semantics for epistemic theories}
\label{sec:g91}

In this section we provide a straightforward generalisation of G91 allowing its application to arbitrary modal theories.
Formulas are extended with the necessity operator $\bL$ according to the following grammar:
\[
\fF ::= \bot \mid a \mid \fF_1 \wedge \fF_2 \mid \fF_1 \vee \fF_2 \mid \fF_1 \to \fF_2 \mid \bL \varphi
\qquad\text{ for any atom } a \in \at. 
\]
An \emph{(epistemic) theory} is a set of formulas as defined above.
%
%
In our context, the epistemic reading of $\bL \psi$ is that ``$\psi$ is one of the agent's beliefs.''
Thus, a formula $\varphi$ is said to be \emph{subjective} if all its atom occurrences (having at least one) are in the scope of $\bL$.
Analogously, $\varphi$ is said to be \emph{objective} if $\bL$ does not occur in $\varphi$.
For instance, $\neg\!\bL a \vee \bL b$ is subjective, $\neg a \vee b$ is objective and $\neg a \vee \bL b$ none of the two.

To represent the agent's beliefs we will use a set $\wv$ of propositional interpretations.
We call \emph{belief set} to each element $I \in \wv$ and \emph{belief view} to the whole set~$\wv$.
The difference between belief and knowledge is that the  former may not hold in the real world.
Thus, satisfaction of formulas will be defined with respect to an interpretation $I \subseteq \at$, possibly $I \not\in \wv$, that accounts for the real world: the pair $\kdint{I}{\wv}$ is called \emph{belief interpretation} (or interpretation in modal logic KD45).
Modal satisfaction is also written \mbox{$\kdint{I}{\wv} \models \varphi$} (ambiguity is removed by the interpretation on the left) and follows the conditions:
\begin{itemize}[ topsep=2pt]
\item $\kdint{I}{\wv} \not\models \bot$,
\item $\kdint{I}{\wv} \models a$ iff $a \in I$, for any atom $a \in \at$,
\item $\kdint{I}{\wv} \models \psi_1 \wedge \psi_2$ iff $\kdint{I}{\wv} \models \psi_1$ and $\kdint{I}{\wv} \models \psi_2$,
\item $\kdint{I}{\wv} \models \psi_1 \vee \psi_2$ iff $\kdint{I}{\wv} \models \psi_1$ or $\kdint{I}{\wv} \models \psi_2$,
\item $\kdint{I}{\wv} \models \psi_1 \to \psi_2$ iff $\kdint{I}{\wv} \not\models \psi_1$ or $\kdint{I}{\wv} \models \psi_2$, and
\item $\kdint{I}{\wv} \models \bL \psi$ iff $\kdint{J}{\wv} \models \psi$ for all $J \in \wv$.
\end{itemize}
Notice that implication here is classical, that is, $\varphi \to \psi$ is equivalent to $\neg \varphi \vee \psi$ in this context.
A belief interpretation $\kdint{I}{\wv}$ is a \emph{belief model} of $\Gamma$ iff $\kdint{J}{\wv} \models \varphi$ for all $\varphi \in \Gamma$ and all $J \in \wv \cup \set{I}$
--  additionally, when $I \in \wv$, we further say that $\wv$ is an  \emph{epistemic model} of $\Gamma$ and abbreviate this as $\wv \models \Gamma$.
%
%
Belief models defined in this way correspond to modal logic KD45 whereas epistemic models correspond to S5.

\begin{example}\label{ex:s5models}
Take the theory $\Gamma_{\ref{f:loop1}}=\{\neg\!\bL b \to a\}$ corresponding to rule \eqref{f:loop1}.
An epistemic model $\wv \models \Gamma_{\ref{f:loop1}}$ must satisfy: 
 $\tuple{\wv,J} \models \bL b$ or $\tuple{\wv,J} \models a$, for all $J \in \wv$.
We get three epistemic models from $\bL b$, $\sset{\{b\}}$, $\sset{\{a,b\}}$, and  $\sset{\{b\},\{a,b\}}$ and the rest of cases must force $a$ true, so we also get $\sset{\{a\}}$ and $\sset{\{a\},\{a,b\}}$.
In other words, $\Gamma_{\ref{f:loop1}}$ has the same epistemic models than $\bL b \vee \bL a$.
\qed
\end{example}
Note that rule \eqref{f:loop1} alone did not seem to provide any reason for believing $b$, but we got three epistemic models above satisfying $\bL b$.
Thus, we will be interested only in some epistemic models (we will call \emph{world views}) that minimize the agent's beliefs in some sense.
To define such a minimisation we rely on the following syntactic transformation provided by~\cite{Truszczynski11}.

\begin{definition}[Subjective reduct]
The \emph{subjective reduct} of a theory~$\Gamma$ with respect to a belief view~$\wv$, also written $\Gamma^\wv$, is obtained by replacing each maximal subformula of the form $\bL \varphi$ by: $\top$, if $\wv \models \bL \varphi$; by $\bot$, otherwise.
Notice that $\Gamma^\wv$ is a classical, non-modal theory.\qed
\end{definition}
Finally, we impose a fixpoint condition where, depending on whether each belief set $I \in \wv$ is required to be a stable model of the reduct or just a classical model, we get G91 or AEL semantics, respectively.

\begin{definition}[AEL and G91 world views]
A belief view~$\wv$ is called an \emph{\mbox{AEL-world} view} of a theory 
$\Gamma$ iff
$\wv=\CL[\Gamma^\wv]$,
and is called a 
\emph{G91-world view} of $\Gamma$ iff
\mbox{$\wv=\SM[\Gamma^\wv]$}.\qed
\end{definition}

\begin{example}[Example~\ref{ex:s5models} revisited]
Take any $\wv$ such that $\wv \models \bL b$.
Then, $\Gamma_{\ref{f:loop1}}^\wv = \{\bot \to a\}$ with $\CL[\Gamma_{\ref{f:loop1}}^\wv]=[\emptyset,\{a\},\{b\},\{a,b\}]$ and $\SM[\Gamma_{\ref{f:loop1}}^\wv]=[\emptyset]$.
None of the two satisfy $\bL b$ so $\wv$ cannot be fixpoint for G91 or AEL.
If $\wv \not\models \bL b$ instead, we get $\Gamma_{\ref{f:loop1}}^\wv = \{\top \to a\}$,
whose classical models are $\{a\}$ and $\{a,b\}$, but only the former is stable.
As a result, $\wv=\sset{\{a\},\{a,b\}}$ is the unique AEL world view and $\wv=\sset{\{a\}}$ the unique G91 world view.
\qed
\end{example}

\begin{example}\label{ex:self-supporting.rule}
Take now the theory $\Gamma_{\ref{f:self}}=\{\bL a \to a\}$ corresponding to rule \eqref{f:self}.
If $\wv \models \bL a$ we get $\Gamma_{\ref{f:self}}^\wv = \set{ \top \to a }$ and $\CL[\Gamma_{\ref{f:self}}^\wv]=\SM[\Gamma_{\ref{f:self}}^\wv]=\{a\}$ so $\wv=\sset{\{a\}}$ is an AEL and G91 world view.
If $\wv \not\models \bL a$, the reduct becomes $\Gamma_{\ref{f:self}}^\wv = \set{ \bot \to a }$, a classical tautology with unique stable model $\emptyset$.
As a result, $\wv=\sset{\emptyset,\{a\}}$ is the other AEL world view, while $\wv=\sset{\emptyset}$ is the a second G91 world view.
\qed
\end{example}

As we can see, the difference between AEL and G91 is that we use classical $\CL[\Gamma^\wv]$ instead of stable $\SM[\Gamma^\wv]$ models, respectively.
It is well known that adding the excluded middle axiom $a \vee \neg a$ for all atoms makes equilibrium logic collapse into classical logic. 
This leads us to the following result.
\begin{theorem}
$\wv$ is an AEL world view of some theory~$\Gamma$ iff $\wv$ is a G91-world view of $\Gamma \cup \{a \vee \neg a \mid a \in \at\}$.\qed
\end{theorem}

\section{Founded world views of epistemic specifications}
\label{sec:unfounded}

As we explained in the introduction, world view $\sset{\{a\}}$ of $\{\bL a \to a\}$ is considered to be ``self-supported'' in the literature but, unfortunately, there is no formal definition for such a concept, to the best of our knowledge.
To cover this lack, we proceed to extend here the idea of unfounded sets from disjunctive logic programs to the epistemic case.
For this purpose, we focus next on the original language of \emph{epistemic specifications}~\cite{Gelfond91} (a fragment of epistemic theories closer to logic programs) on which most approaches have been actually defined.

Let us start by introducing some terminology.
An \emph{objective literal} is either an atom~$a \in \at$, its negation~$\neg a$ or its double negation~$\neg\neg a$.
A \emph{subjective literal} is any of the formulas\footnote{We focus here on the study of the operator~$\bL$, but epistemic specifications also allow a second operator $\bM l$
whose relation to $\bL$ is also under debate and, for this reason, we leave it future work.} $\bL l$, $\neg\! \bL l$ or $\neg \neg\! \bL l$ where $l$ an objective literal.
A \emph{literal} is either an objective or a subjective literal, and is called \emph{negative} if it contains negation and \emph{positive} otherwise.
A \emph{rule} is a formula of the form
\begin{gather}
a_1 \vee \dotsc \vee a_n \leftarrow B_1 \wedge \dotsc \wedge B_m \label{f:rule}
\end{gather}
with $n\geq 0$, $m \geq 0$ and $m+n>0$, where each $a_i$ is an atom and each $B_j$ is a literal.
For any rule $r$ like \eqref{f:rule}, we define its body as $\Body(r) \eqdef B_1 \wedge \dotsc \wedge B_m$ and its head $\Head(r) \eqdef a_1 \vee \dotsc \vee a_n$, which we sometimes use as the set of atoms $\{a_1,\dotsc,a_n\}$.
When $n=0$, $\Head(r)=\bot$ and the rule is a \emph{constraint}, whereas if $m=0$ then $\Body(r)=\top$ and the rule is a \emph{fact}.
%
The set $\Bodyrp(r)$ collects all atoms occurring in positive objective literals in the body while $\Bodymp(r)$ collects all atoms occurring in positive subjective literals. 
An \emph{epistemic specification} or \emph{program} is a set of rules.
As with formulas, a program without occurrences of~$\bL$ is said to be objective (it corresponds to a standard disjunctive logic program with double negation).

\begin{definition}[Unfounded set]\label{def:unfoundedset}
Let $\Pi$ be a program and $\wv$ a belief view.
An \emph{unfounded set} $\us$ with respect to $\Pi$ and $\wv$ is a non-empty set of pairs where, for each $\tuple{X,I} \in \us$, we have that $X$ and $I$ are sets of atoms and there is no rule $r \in \Pi$ with $\Head(r) \cap X \neq \emptyset$ satisfying:
\begin{enumerate}[itemsep=3pt,topsep=2pt]
\item $\kdint{I}{\wv} \models \Body(r)$
	\label{item:1:def:unfounded}
\item $\Bodyrp(r) \cap X = \emptyset$
	\label{item:2:def:unfounded}
\item $(\Head(r) \setminus X) \cap I = \emptyset$
	\label{item:3:def:unfounded}
\item $\Bodymp(r) \cap Y = \emptyset$ with $Y = \bigcup \setm{ X' }{ \tuple{X',I'} \in \us }$.\qed
	\label{item:4:def:unfounded}
\end{enumerate} 
\end{definition}
The definition works in a similar way to standard unfounded sets~\cite[Definition~3.1]{LeoneRS97}.
In fact, the latter corresponds to the first three conditions above, except that we use $\kdint{I}{\wv}$ to check $\Body(r)$, as it may contain now subjective literals.
Intuitively, each $I$ represents some potential belief set (or stable model) and $X$ is some  set of atoms without a ``justifying'' rule, that is, there is no $r \in \Pi$ allowing a positive derivation of atoms in $X$.
A rule like that should have a true $\Body(r)$ (condition~\ref{item:1:def:unfounded}) but not because of positive literals in $X$ (condition~\ref{item:2:def:unfounded}) and is not used to derive other head atoms outside $X$ (condition~\ref{item:3:def:unfounded}).
The novelty in our definition is the addition of condition~\ref{item:4:def:unfounded}: to consider $r$ a justifying rule, we additionally require not using any positive literal $\bL a$ in the body such that atom $a$ also belongs to any of the unfounded components $X'$ in $\us$.

\begin{definition}[Founded world view]\label{def:unfounded}
Let $\Pi$ be a program and $\wv$ be a belief view.
We say that $\wv$ is \emph{unfounded} if there is some 
unfounded-set~$\us$ s.t.,
for every $\tuple{X,I} \in \us$,
we have $I \in \wv$ and $X \cap I \neq \emptyset$.
$\wv$ is called \emph{founded} otherwise.\qed
\end{definition}

When $\Pi$ is an objective program, each pair $\tuple{X,I}$ corresponds to a standard unfounded set $X$ of some potential stable model~$I$ in the traditional sense of~\cite{LeoneRS97}.

\begin{example}\label{ex:or}
Given the single disjunctive rule $a \vee b$ suppose we check the (expected) world view $\wv=\sset{\{a\},\{b\}}$.
For $I=\{a\}$ and $X=\{a\}$, rule $a \vee b$ satisfies the four conditions and justifies $a$.
The same happens for $I=\{b\}=X$.
So, $\wv$ is founded.
However, suppose we try with $\wv'=\sset{\{a,b\}}$ instead. 
For $I=\{a,b\}$ we can form $X=\{a\}$ and $X'=\{b\}$ and in both cases, the only rule in the program, $a \vee b$, violates condition~\ref{item:3:def:unfounded}.
As a result, $\wv'$ is unfounded due to the set $\us'=\{\tuple{\{b\},\{a,b\}},\tuple{\{a\},\{a,b\}}\}$.\qed
\end{example}

To illustrate how condition 4 works, let us continue with Example~\ref{ex:self-supporting.rule}.
\begin{examplecont}{ex:self-supporting.rule}\label{ex:self-supporting.rule2}
Theory $\Gamma_{\ref{f:self}}=\{\bL a \to a\}$ is also a program.
Given belief set $\wv=\sset{\{a\}}$ we can observe that $\us=\sset{\tuple{\set{a},\set{a}}}$ makes $\wv$ unfounded because the unique rule in $\Gamma_{\ref{f:self}}$ does not fulfill condition 4: we cannot derive $a$ from a rule that contains $a \in \Bodymp(r)$.
On the other hand, the other G91 world view, $\wv=\sset{\emptyset}$, is trivially founded.\qed
\end{examplecont}

%

Since Definition~\ref{def:unfounded} only depends on some epistemic program and its selected world views, we can raise it to a general property for any epistemic semantics.

\begin{property}[Foundedness]\label{property:unfounded-freedom}
A semantics satisfies \emph{foundedness} when all the world views it assigns to any program $\Pi$ are founded. \qed
\end{property}

Approaches proposed after G91 do remove unfounded world views in the examples studied in the literature, but unfortunately, this does not mean that they generally satisfy foundedness.
Let us consider a common counterexample.
\begin{example}\label{ex:unfounded}
Take the epistemic logic program\newprogramhide\label{theory:larger.set.of.worlds}:
\begin{gather*}
a \vee b
\hspace{2cm}
a \leftarrow \bL b
\hspace{2cm}
b \leftarrow \bL a
\tag{$\program\ref{theory:larger.set.of.worlds}$}
\end{gather*}
whose G91-world views are
$\wv = \sset{ \set{a}, \set{b} }$
and
$\wv' = \sset{ \set{a , b} }$.
These are, indeed, the two
cases we analysed in Example~\ref{ex:or}.
$\wv$ is again founded because $a \vee b$ keeps justifying both possible $\tuple{X,I}$ pairs, that is, $\sset{\tuple{\set{a},\set{a}}}$ and $\sset{\tuple{\set{b},\set{b}}}$.
However, for $\wv'$ we still have the unfounded set $\us' = \sset{ \tuple{\set{a},\set{a,b}}, \, \tuple{\set{b}, \set{a,b}}}$
which violates condition~\ref{item:3:def:unfounded} for the first rule as before, but also condition~\ref{item:4:def:unfounded} for the other two rules.\qed
\end{example}
Note how $\us'$ allows us to spot the root of the derivability problem: to justify $a$ in $\tuple{\set{a},\set{a,b}}$ we cannot use $a \leftarrow \bL b$ because $b$ is part of the unfounded structure $X$ in the other pair $\tuple{\{b\},\{a,b\}}$, and vice versa.
Since the variants by Gelfond in~\cite{Gelfond11} (G11) and Kahl et al.~\cite{Kahl15} (K15) also assign the unfounded world view $\wv'$ to $\program\ref{theory:larger.set.of.worlds}$ (in fact, they coincide with G91 for this program), we can conclude that G11 and K15 \emph{do not satisfy foundedness} either.

A more elaborated strategy is adopted by the recent approaches by Fari\~nas et al. ~\cite{CerroHS15} (F15) and Shen and Eiter~\cite{ShenE17} (S17), that treat the previous world views as candidate solutions\footnote{In~\cite{CerroHS15}, these candidate world views are called \emph{epistemic equilibrium models} while selected world views receive the name of \emph{autoepistemic equilibrium models}.}, but select the ones with \emph{minimal knowledge} in a second step.
This allows removing the unfounded world view $\sset{\{a,b\}}$ in Example~\ref{ex:unfounded}, because the other solution $\sset{\{a\},\{b\}}$ provides less knowledge.
Unfortunately, this strategy does not suffice to guarantee foundedness, since other formulas (such as  constraints) may remove the founded world view, as explained below.

\begin{examplecont}{ex:unfounded}\label{ex:unfounded2}
Take the program 
\mbox{$\newprogram\label{theory:larger.set.of.worlds2} = \program\ref{theory:larger.set.of.worlds} \cup \set{ \bot \leftarrow \neg\!\bL a}$}.
The constraint rules out $\wv=\sset{\{a\},\{b\}}$ because the latter satisfies $\neg\!\bL a$.
In  G91, G11, F15 and S17, only world view $\wv'=\sset{\{a,b\}}$ is left, so knowledge minimisation has no effect.
However, $\wv'$ is still unfounded in $\program\ref{theory:larger.set.of.worlds2}$ since constraints do not affect that feature (their empty head never justifies any atom).\qed
\end{examplecont}

As a conclusion, semantics F15 and S17 \emph{do not satisfy foundedness} either.

\section{Founded Autoepistemic Equilibrium Logic}\label{sec:autoepistemic.equilibrium.logic}

We present now the semantics proposed in this paper, introducing \emph{Founded Autoepistemic Equilibrium Logic} (FAEEL).
The basic idea is an elaboration of the belief (or KD45) interpretation $\kdint{I}{\wv}$ already seen but replacing belief sets by HT pairs.
Thus, we extend now the idea of \emph{belief view} $\wv$ to a non-empty set of HT-interpretations $\wv=\{\tuple{H_1,T_1}, \dots, \tuple{H_n,T_n}\}$ and  say that $\wv$ is \emph{total} when $H_i=T_i$ for all of them, coinciding with the form of belief views $\wv=\{T_1, \dots, T_n\}$ we had so far.
Similarly, a \emph{belief interpretation} is now redefined as $\kdint{\tuple{H,T}}{\wv}$, or simply $\kdint{H,T}{\wv}$, where $\wv$ is a belief view and $\tuple{H,T}$ stands for the real world, possibly not in $\wv$.
Next, we redefine the satisfaction relation from a combination of modal logic KD45 and HT.
A belief interpretation $\cI =ÃÂ \kdint{H,T}{\wv}$ satisfies a formula $\varphi$, written $\cI \models \varphi$, iff:
\begin{itemize}[topsep=2pt]
\item $\cI \not\models \bot$,
\item $\cI \models a$ iff $a \in H$, for any atom $a \in \at$,
\item $\cI \models \psi_1 \wedge \psi_2$ iff $\cI \models \psi_1$ and $\cI \models \psi_2$,
\item $\cI \models \psi_1 \vee \psi_2$ iff $\cI \models \psi_1$ or $\cI \models \psi_2$,
\item $\cI \models \psi_1 \to \psi_2$ iff both: (i) $\cI \not\models \psi_1$ or $\cI \models \psi_2$;
and (ii) $\kdint{T}{\wv^t} \not\models \psi_1$ or $\kdint{T}{\wv^t} \models \psi_2$, where $\wv^t=\{T_i \mid \tuple{H_i,T_i} \in \wv\}$.
\item $\cI \models \bL \psi$ iff $\kdint{H_i,T_i}{\wv} \models \psi$ for all $\tuple{H_i,T_i} \in \wv$.
\end{itemize}
For total belief interpretations, this new satisfaction relation collapses to the one in Section~\ref{sec:g91} (that is, KD45).
Interpretation $\kdint{H,T}{\wv}$ is a \emph{belief model} of $\Gamma$ iff $\kdint{H_i,T_i}{\wv} \models \varphi$ for all $\tuple{H_i,T_i} \in \wv \cup\{\tuple{H,T}\}$ and all $\varphi \in \Gamma$
 -- 
additionally, when $\tuple{H,T} \in \wv$, we further say that $\wv$ is an \emph{epistemic model} of  $\Gamma$, abbreviated as $\wv \models \Gamma$.

\begin{Proposition}[Persistence]{\label{prop:persistance}}
$\kdint{H,T}{\wv}\models \varphi$ implies $\kdint{T}{\wv^t} \models \varphi$.\qed
\end{Proposition}

A belief model just captures collections of HT models which need not be in equilibrium.
To make the agent's beliefs correspond to stable models we impose a particular minimisation criterion on belief models.

\begin{definition}\label{def:int.prec}
We define the partial order $\cI' \preceq \cI$ for belief interpretations $\cI'=\kdint{H',T'}{\wv'}$ and $\cI=\kdint{H,T}{\wv}$ when the following three conditions hold:

\begin{enumerate}
\item $T' = T$ and $H' \subseteq H$, and

\item for every $\tuple{H_i,T_i} \in \wv$,
there is some $\tuple{H'_i,T_i} \in \wv'$,
with $H'_i \subseteq H_i$.

\item for every $\tuple{H'_i,T_i} \in \wv'$,
there is some $\tuple{H_i,T_i} \in \wv$,
with $H'_i \subseteq H_i$.\qed
\end{enumerate}
\end{definition}
As usual, $\cI' \prec \cI$ means $\cI' \preceq \cI$ and $\cI' \neq \cI$.
The intuition for $\cI' \preceq \cI$ is that $\cI'$ contains less information than $\cI$ for each fixed $T_i$ component.
As a result, $\cI' \models \varphi$ implies $\cI \models \varphi$ for any formula $\varphi$ without implications other than $\neg \psi = \psi \to \bot$.

\begin{definition}
A total belief interpretation $\cI=\kdint{T}{\wv}$ is said to be an \emph{equilibrium belief model} of some theory~$\Gamma$ iff $\cI$ is a belief model of $\Gamma$ and there is no other belief model $\cI'$ of $\Gamma$ such that $\cI' \prec \cI$.\qed
\end{definition}
By $\MEQ[\Gamma]$ we denote the set of equilibrium belief models of $\Gamma$.
As a final step, we impose a fixpoint condition to minimise the agent's knowledge as follows.
\begin{definition}\label{def:au-eqmodel}
A belief view~$\wv$ is called an \emph{equilibrium world view} of~$\Gamma$
iff:
\begin{IEEEeqnarray*}{c+x*}
\wv \ \ = \ \ \setm{ T }{ \kdint{T}{\wv } \in \MEQ[ \Gamma ] }
&\qed
\end{IEEEeqnarray*}
\end{definition}

\begin{examplecont}{ex:self-supporting.rule2}\label{ex:self-supporting.rule3}
Back to~\mbox{$\Gamma_{\ref{f:self}}=\{\bL a \to a\}$}, remember its unique founded G91-world view was~$\sset{\emptyset}$.
It is easy to see that $\cI = \kdint{\emptyset}{\sset{\emptyset}} \in \MEQ[\Gamma_{\ref{f:self}}]$ because $\kdint{\emptyset}{\sset{\emptyset}} \models \Gamma_{\ref{f:self}}$ and no smaller belief model can be obtained.
Moreover, $\sset{\emptyset}$ is an equilibrium world view of~$\Gamma_{\ref{f:self}}$ since no other
\mbox{$T \not\in \sset{\emptyset}$}
satisfies
\mbox{$\kdint{T}{\sset{\emptyset}} \in \MEQ[\Gamma_{\ref{f:self}}]$}.
The only possibility is $\kdint{\set{a}}{\sset{\emptyset}}$ but it fails because there is a smaller belief model
$\kdint{\emptyset,\set{a}}{\sset{\emptyset}}$
satisfying \mbox{$\bL a \to a$}.
As for the other potential world view $\sset{\set{a}}$, it is not in equilibrium: we already have $\cI' = \kdint{\set{a}}{\sset{\set{a}}} \not\in \MEQ[\Gamma_{\ref{f:self}}]$
because the smaller interpretation $\cI ''=\kdint{\set{a},\set{a}}{\sset{\tuple{\emptyset,\set{a}}}}$ also satisfies $\Gamma_{\ref{f:self}}$.
In particular, note that $\cI'' \not\models \bL a$
and, thus, clearly satisfies \mbox{$\bL a \to a$.
\qed}
\end{examplecont}

The logic induced by equilibrium world views is called \emph{Founded Autoepistemic  Equilibrium Logic} (FAEEL).
A first important property is:
\begin{Theorem}{\label{thm:unfounded-free}}
FAEEL satisfies foundedness.\qed
\end{Theorem}

A second interesting feature is that equilibrium world views are also G91-world views though the converse may not be the case (as we just saw in Example~\ref{ex:self-supporting.rule3}).
This holds, not only for programs, but in general for any  theory:

\begin{Theorem}{\label{thm:g91}}
For any theory $\Gamma$, its equilibrium world views are also G91-world views of $\Gamma$.\qed
\end{Theorem}

In other words, FAEEL is strictly stronger than G91, something that, as we see next, is not the case in other approaches in the literature.

\begin{example}\label{ex:nog91}
The following program\newprogramhide\label{prg:nog91}:
\begin{gather*}
a \vee b
\hspace{2cm}
c \leftarrow \bL a
\hspace{2cm}
\bot \leftarrow \neg c
	\tag{\program\ref{prg:nog91}}
\end{gather*}
has no G91-world views, but according to G11, K15, F15 and S17 has world view~$\sset{\set{a,c}}$.
This example was also used in~\cite{CabalarFF2018} to show that these semantics do not satisfy another property, called there \emph{epistemic splitting}.\qed
\end{example}

\begin{examplecont}{ex:unfounded}\label{ex:unfounded3}
Take again program~$\program\ref{theory:larger.set.of.worlds}$ whose G91-world views were $\wv=\sset{\set{a},\set{b}}$ and $\wv'=\sset{\set{a, b}}$.
Since $\wv'$ is unfounded, it cannot be an equilibrium world view (Theorem~\ref{thm:unfounded-free}), leaving $\wv$ as the only candidate (Theorem~\ref{thm:g91}).
Let us check that this is in fact an equilibrium world view.
First, note that
$\cI = \kdint{\set{a}}{\sset{\set{a},\set{b}}} \in \MEQ[\program\ref{theory:larger.set.of.worlds}]$ 
because there is no model $\cI'$ of $\program\ref{theory:larger.set.of.worlds}$
such that $\cI' \preceq \cI$.
In fact, it is easy to see that
$\kdint{H_3,\set{a}}{\sset{\tuple{H_1,\set{a}},\tuple{H_2,\set{b}}}}$ is not a model of the rule $a \vee b$ if $H_i = \emptyset$ for any $i \in \set{1,2,3}$.
Symmetrically, we have that
$\cI' = \kdint{\set{b}}{\sset{\set{a},\set{b}}} \in \MEQ[\program\ref{theory:larger.set.of.worlds}]$ too.
Finally, we have to check that no other $T \not\in \sset{\{a\},\{b\}}$ can form an equilibrium belief model.
For the case $T=\emptyset$, it is easy to check that
$\kdint{\emptyset}{\sset{\set{a},\set{b}}}$
does not satisfy $a \vee b$.
For $T=\{a,b\}$,
we have that
$\cI'' = \kdint{\set{a,b}}{\sset{\set{a},\set{b}}} \not\in \MEQ[\program\ref{theory:larger.set.of.worlds}]$ because, for instance, the smaller $\cI''' = \kdint{\set{a},\set{a,b}}{\sset{\set{a},\set{b}}}$  is a model of $\program\ref{theory:larger.set.of.worlds}$.
\qed
\end{examplecont}

Theorems~\ref{thm:unfounded-free} and~\ref{thm:g91} assert that any equilibrium world view is a founded G91-world view.
The natural question is whether the opposite also holds.
In Examples~\ref{ex:self-supporting.rule3}, \ref{ex:nog91} and~\ref{ex:unfounded3} we did not find any counterexample, and this is in fact a general property, as stated below.

\begin{mtheorem}
Given any program~$\Pi$, its equilibrium world views coincide with its founded G91-world views.\qed
\end{mtheorem}

\begin{figure}[t]
\small
\begin{tabular}{ | C{1.75cm} | C{2cm}  | }
\hline
program & world views
\\\hline
\multirow{2}{*}{$a \vee b$} &   \multirow{2}{*}{ \sset{\set{a} , \set{b}} }
\\
&
\\\hline
$a \vee b$ &  \multirow{2}{*}{\sset{\set{a} , \set{b}} }
\\
$a \leftarrow \bK b$ &
\\\hline
$a \vee b$ & \multirow{2}{*}{\sset{\set{a} }  }
\\
$a \leftarrow \neg\!\bK b$ & 
\\\hline
$a \vee b$ & \multirow{2}{*}{\sset{\set{a,c}, \set{b,c} } }
\\
$c \leftarrow \neg\!\bK b$ & 
\\\hline
$a \leftarrow \neg\!\bK b$ & \multirow{2}{*}{\sset{\set{a}} , \sset{\set{b} } }
\\
$b \leftarrow \neg\!\bK a$ & 
\\\hline
$a \leftarrow \neg\!\bK\!\neg a$ &  \multirow{2}{*}{\sset{\set{a}}  }
\\
$a \leftarrow \neg\!\bK a$ &
\\\hline\multicolumn{2}{l}{}
\\\multicolumn{2}{l}{}
\\
\end{tabular}
\hspace{0.6cm}
\begin{tabular}{ | C{2cm} | C{2.8cm} | C{2.2cm} | }
\hline
program & G91/G11/FAEEL & K15/F15/S17
\\\hline
\multirow{2}{*}{
$a \leftarrow \neg\!\bK\!\neg a$} & 
\multirow{2}{*}{\sset{\emptyset } , \sset{\set{a}} }
	&
\multirow{2}{*}{\sset{\set{a} }}
\\&&
\\\hline
$a \vee b$ & \multirow{2}{*}{ \text{none} }
&
 \multirow{2}{*}{\sset{\set{a} } } 
\\
$a \leftarrow \neg\!\bK\!\neg b$ & 
 
&
 
\\\hline
$a \vee b$ &  \multirow{2}{*}{ \sset{\set{a}} ,  \sset{\set{a} , \set{b} } }
&
 \multirow{2}{*}{ \sset{\set{a} , \set{b} } } 
\\
$a \leftarrow \bK\!\neg b$ &   
&

\\\hline
$a \leftarrow b$ &
 \multirow{2}{*}{\sset{\emptyset} , \sset{\set{a,b}}  } 
&
 \multirow{2}{*}{\sset{\set{a,b}}  } 
\\
$b \leftarrow \neg\!\bK\!\neg a$ &

&

\\\hline
$a \leftarrow \neg\!\bK\!\neg b$ &  
\multirow{2}{*}{\sset{\emptyset} , \sset{\set{a} , \set{b} } }
&
\multirow{2}{*}{\sset{\set{a} , \set{b} } }
\\
$b \leftarrow \neg\!\bK\!\neg a$ &
&
\\\hline\multicolumn{3}{c}{}
\\\multicolumn{3}{c}{}
\\\multicolumn{3}{c}{}
\\\multicolumn{3}{c}{}
\end{tabular}
\vspace{-1cm}

\caption{On the left, examples where G91, G11, K15, F15, S17 and FAEEL agree. On the right, examples where FAEEL/G91/G11 differ from K15/F15/S17.}\label{table:agreement}
\end{figure}

An interesting observation is that in all the original examples of epistemic specifications~\cite{Gelfond91,GelfondP93} used by Gelfond to introduce G91, modal operators occurred in the scope of negation.
Negated beliefs never incur unfoundedness, so this feature could not be spotted using this family of examples.
In fact, under this syntactic restriction, FAEEL and G91 coincide.

\begin{Proposition}{\label{prop:neg-L}}
For any theory where all occurrences of $\bL$ are in the scope of negation, we have that the equilibrium world views and the G91-world views coincide.\qed
\end{Proposition}

Proposition~\ref{prop:neg-L} also holds for semantics~\cite{Truszczynski11,WangZ05} that are conservative extensions of G91, as well as for G11.
Apart from foundedness, \cite{CabalarFF2018} recently proposed other four properties for semantics of epistemic specifications.
We analyse here three of them, omitting the so-called \emph{epistemic splitting} due to lack of space.
\begin{enumerate}
\item 
\emph{supra-ASP} holds when, for any objective theory $\Gamma$, either $\Gamma$ has a unique world view $\wv=\SM[\Gamma] \neq \emptyset$ or $\SM[\Gamma]=\emptyset$ and $\Gamma$ has no world view. 

\item
\emph{supra-S5} holds when every world view $\wv$ of a theory $\Gamma$ is also an \sfmodel of~$\Gamma$ (that is, $\wv \models \Gamma$).

\item
\emph{subjective constraint monotonicity} holds when, for any theory $\Gamma$ and any subjective constraint \mbox{$\bot\leftarrow\varphi$},
we have that $\wv$ is a world view of $\Gamma \cup \{\bot\leftarrow\varphi\}$ iff both $\wv$ is a world view of $\Gamma$ and $\wv$ is not an \sfmodel of~$\varphi$.
\end{enumerate}

\begin{Proposition}{\label{prop:properties}}
FAEEL satisfies supra-ASP, 
supra-S5 and
subjective constraint monotonicity.\qed
\end{Proposition}
All semantics discussed in this paper satisfy the above first two properties but most of them fail for subjective constraint monotonicity, as first discussed in~\cite{KahlL18}.
In fact, a variation of Example~\ref{ex:nog91} can be used to show that K15, F15 and S17 do not satisfy 
this property.





%

\begin{examplecont}{ex:nog91}\label{ex:nog91b}
Suppose we remove the constraint (last rule) from~\program\ref{prg:nog91} getting the program\newprogramhide\label{prg:nog91b}
$\program\ref{prg:nog91b} = \set{ a \vee b
\,,\,
c \leftarrow \bL a}$.
%
All semantics, including G91 and FAEEL, agree that~\program\ref{prg:nog91b} has a unique world view $\sset{\set{a},\set{b}}$.
Suppose we add now a subjective constraint 
$\newprogram\label{prg:nog91c} = \program\ref{prg:nog91b} \cup \set{ \bot \leftarrow \neg\!\bL c }$.
This addition leaves G91 and FAEEL without world views (due to subjective constraint monotonicity)
the same happens for G11, but not for K15, F15 and S17, which provide a \emph{new} world view~$\sset{\set{a,c}}$ not obtained before adding the subjective constraint.\qed
\end{examplecont}

Tables~\ref{table:agreement} and~\ref{table:disagreement} show a list of examples taken from Table~4 in~\cite{CerroHS15} and their world views according to different semantics.
\begin{figure}[t]
\small
\begin{tabular}{ | C{3.8cm} | C{1.95cm} | C{2.1cm} | C{1.85cm} | C{1.9cm} | }
\hline
program & G91 & G11/FAEEL & K15 & F15/S17
\\\hline
$a \leftarrow \neg\!\bK\!\neg b \wedge \neg b$ &  \multicolumn{3}{c|}{\multirow{2}{*}{\sset{\emptyset} , \sset{\set{a} , \set{b} } }}
& \multirow{2}{*}{$\sset{ \set{a}, \set{b} }$}
\\
$b \leftarrow \neg\!\bK\!\neg a \wedge \neg a$ & \multicolumn{3}{c|}{}
&
\\\hline
$a \leftarrow \bK a$ & 
	\sset{\emptyset } , \sset{\set{a}}
	&
	\multicolumn{3}{c|}{ \sset{\emptyset } }
\\\hline
$a \leftarrow \bK a$ & \multirow{2}{*}{ \sset{\set{a}} }
	& \multicolumn{3}{c|}{ \multirow{2}{*}{  \text{none} } }
\\
$a \leftarrow \neg\!\bK a$ & & \multicolumn{3}{c|}{}
\\\hline
\end{tabular}
\caption{Examples splitting different semantics.
Examples~\ref{ex:unfounded} and~\ref{ex:nog91} in the paper can be used to further split FAEEL and G11.}\label{table:disagreement}
\end{figure}


\section{Conclusions}\label{sec:conclusions}

In order to characterise self-supported world-views, already present in Gelfond's 1991 semantics~\cite{Gelfond91} (G91), we have extended the definition of unfounded sets from standard logic programs to epistemic specifications.
As a result, we proposed the \emph{foundedness} property for epistemic semantics, which is not satisfied by other approaches in the literature.
Our main contribution has been the definition of a new semantics, based on the so-called \emph{Founded Autoepistemic Equilibrium Logic} (FAEEL), that satisfies foundedness.
This semantics actually covers the syntax of any arbitrary modal theory and is a combination of Equilibrium Logic and Autoepistemic Logic.
As a main result, we were able to prove that, for the syntax of epistemic specifications, FAEEL world views coincide with the set of G91 world views that are founded.
We showed how this semantics behaves on a set of common examples in the literature and proved that it satisfies other three basic properties: all world views are S5 models (\emph{supra-S5}); standard programs have (at most) a unique world view containing all the stable models (\emph{supra-ASP}); and subjective constraints just remove world views (monotonicity).
FAEEL also satisfies the property of~\emph{epistemic splitting} as proposed in~\cite{CabalarFF2018}, but we leave the proof and discussion for future work, together with a formal comparison with other approaches.

\bibliography{refs}

\begin{thebibliography}{10}
\providecommand{\url}[1]{\texttt{#1}}
\providecommand{\urlprefix}{URL }

\bibitem{CabalarFF2018}
{Cabalar}, P., {Fandinno}, J., {Fari{\~n}as del Cerro}, L.: {Splitting
  Epistemic Logic Programs}. arXiv e-prints arXiv:1812.08763 (Dec 2018)

\bibitem{CerroHS15}
{Fari{\~{n}}as~del~Cerro}, L., Herzig, A., Su, E.I.: Epistemic equilibrium
  logic. In: Proc. of the Intl. Joint Conference on Artificial Intelligence
  (IJCAI'15). pp. 2964--2970. {AAAI} Press (2015)

\bibitem{GelderRS91}
Gelder, A.V., Ross, K.A., Schlipf, J.S.: The well-founded semantics for general
  logic programs. J. {ACM}  38(3),  620--650 (1991)

\bibitem{Gelfond91}
Gelfond, M.: Strong introspection. In: Dean, T.L., McKeown, K. (eds.)
  Proceedings of the {AAAI} Conference. vol.~1, pp. 386--391. {AAAI} Press/The
  {MIT} Press (1991)

\bibitem{Gelfond11}
Gelfond, M.: New semantics for epistemic specifications. In: {LPNMR}. Lecture
  Notes in Computer Science, vol. 6645, pp. 260--265. Springer (2011)

\bibitem{GL88}
Gelfond, M., Lifschitz, V.: The stable model semantics for logic programming.
  In: Proc. of the 5th Intl. Conference on Logic Programming (ICLP'88). pp.
  1070--1080 (1988)

\bibitem{GelfondP93}
Gelfond, M., Przymusinska, H.: Reasoning on open domains. In: Proc. of the
  Intl. Conf. on Logic Programming and Non-Monotonic Reasoning (LPNMR'93). pp.
  397--413. {MIT} Press (1993)

\bibitem{Hey30}
Heyting, A.: Die formalen {R}egeln der intuitionistischen {L}ogik.
  Sitzungsberichte der Preussischen Akademie der Wissenschaften,
  Physikalisch-mathematische Klasse pp. 42--56 (1930)

\bibitem{Kahl15}
Kahl, P., Watson, R., Balai, E., Gelfond, M., Zhang, Y.: The language of
  epistemic specifications (refined) including a prototype solver. Journal of
  Logic and Computation  (2015)

\bibitem{KahlL18}
Kahl, P.T., Leclerc, A.P.: Epistemic logic programs with world view
  constraints. In: {ICLP} (Technical Communications). {OASICS}, vol.~64, pp.
  1:1--1:17. Schloss Dagstuhl - Leibniz-Zentrum fuer Informatik (2018)

\bibitem{Konolige88}
Konolige, K.: On the relation between default and autoepistemic logic. Artif.
  Intell.  35(3),  343--382 (1988)

\bibitem{LK18}
Lecrerc, A.P., Kahl, P.T.: A survey of advances in epistemic logic program
  solvers. In: Proc. of the 11th Intl. Workshop on Answer Set Programming and
  other Computer Paradigms (ASPOCP'18) (2018)

\bibitem{LeoneRS97}
Leone, N., Rullo, P., Scarcello, F.: Disjunctive stable models: Unfounded sets,
  fixpoint semantics, and computation. Inf. Comput.  135(2),  69--112 (1997)

\bibitem{MT99}
Marek, V., Truszczy{\'n}ski, M.: Stable models and an alternative logic
  programming paradigm, pp. 169--181. Springer-Verlag (1999)

\bibitem{MarekT89}
Marek, V.W., Truszczynski, M.: Relating autoepistemic and default logics. In:
  {KR}. pp. 276--288. Morgan Kaufmann (1989)

\bibitem{Moore84}
Moore, R.C.: Possible-world semantics for autoepistemic logic. In: Proceedings
  of the Non-Monotonic Reasoning Workshop, Mohonk Mountain House, New Paltz,
  {NY} 12561, USA, October 17-19, 1984. pp. 344--354. American Association for
  Artificial Intelligence {(AAAI)} (1984)

\bibitem{Moore85}
Moore, R.C.: Semantical considerations on nonmonotonic logic. Artif. Intell.
  25(1),  75--94 (1985), \url{https://doi.org/10.1016/0004-3702(85)90042-6}

\bibitem{Nie99}
Niemel{\"a}, I.: Logic programs with stable model semantics as a constraint
  programming paradigm. AMAI  25,  241--273 (1999)

\bibitem{Niemela91}
Niemel{\"{a}}, I.: Constructive tightly grounded autoepistemic reasoning. In:
  {IJCAI}. pp. 399--405. Morgan Kaufmann (1991)

\bibitem{Pearce96}
Pearce, D.: A new logical characterisation of stable models and answer sets.
  In: {NMELP}. Lecture Notes in Computer Science, vol. 1216, pp. 57--70.
  Springer (1996)

\bibitem{Schwarz91}
Schwarz, G.: Autoepistemic logic of knowledge. In: Nerode, A., Marek, V.W.,
  Subrahmanian, V.S. (eds.) Logic Programming and Non-monotonic Reasoning,
  Proceedings of the First International Workshop, Washington, D.C., USA, July
  1991. pp. 260--274. The {MIT} Press (1991)

\bibitem{schwarz1992minimal}
Schwarz, G.: Minimal model semantics for nonmonotonic modal logics. In: Logic
  in Computer Science, 1992. LICS'92., Proceedings of the Seventh Annual IEEE
  Symposium on. pp. 34--43. IEEE (1992)

\bibitem{ShenE17}
Shen, Y., Eiter, T.: Evaluating epistemic negation in answer set programming
  (extended abstract). In: Proc. of the Intl. Joint Conference on Artificial
  Intelligence (IJCAI'17). pp. 5060--5064 (2017)

\bibitem{Truszczynski11}
Truszczy{\'n}ski, M.: Revisiting epistemic specifications. In: Logic
  Programming, Knowledge Representation, and Nonmonotonic Reasoning. Lecture
  Notes in Computer Science, vol. 6565, pp. 315--333. Springer (2011)

\bibitem{WangZ05}
Wang, K., Zhang, Y.: Nested epistemic logic programs. In: {LPNMR}. Lecture
  Notes in Computer Science, vol. 3662, pp. 279--290. Springer (2005)

\end{thebibliography}

\newpage
\appendix

\section{Proof of Proposition~\ref{prop:persistance}}

\begin{Proofof}{\ref{prop:persistance}}
Just note that, for atomic $a$, we have that
$\tuple{\wv,H,T} \models a$
iff
$a \in H \subseteq T$
which implies that
$\tuple{\wv^t,T} \models a$.
The rest of the proof follows by induction in the structure of $\varphi$.
\end{Proofof}

\section{Proof of Theorem~\ref{thm:unfounded-free}}

The proof of Theorem~\ref{thm:unfounded-free} rely on the definition of \sfequilibrium models. We first show that \sfequilibrium models are unfounded-free and then that autoepistemic world views are \sfequilibrium models.
Then, we obtain that autoepistemic world views are unfounded-free as a corollary.
We start by defining \sfhtinterpretations~$\wv$ as sets of \htinterpretations.
An \sfhtinterpretations
$\wv$ is said to be \emph{total} iff
it satisfies that every \htinterpretation~$\tuple{H,T} \in \wv$ is total, that is, $H = T$.
We say that an  \sfhtinterpretations~$\wv$ is a \sfhtmodel of a formula~$\varphi$
iff $\tuple{\wv,H,T} \models \varphi$ for every $\tuple{H,T}\in \wv$.

\begin{definition}
Given \sfhtinterpretations
$\wv_1$
and
$\wv_2$,
we write $\wv_1 \preceq \wv_2$
iff the following two condition hold:
\begin{enumerate}
\item for every $\tuple{H_2,T} \in \wv_2$,
there is some $\tuple{H_1,T} \in \wv_1$,
with $H_1 \subseteq H_2$.

\item for every $\tuple{H_1,T} \in \wv_1$,
there is some $\tuple{H_2,T} \in \wv_2$,
with $H_1 \subseteq H_2$.
\end{enumerate}
As usual, we write $\wv_1 \prec \wv_2$ iff $\wv_1 \preceq \wv_2$ and $\wv_1 \neq \wv_2$.\qed
\end{definition}

\begin{definition}
A total \sfhtinterpretation $\wv$ is said to be a \emph{\sfequilibrium model} of some theory~$\Gamma$ iff $\wv$ is a \sfhtmodel of $\Gamma$ and there is no other \sfhtmodel~$\wv'$ of $\Gamma$ such that $\wv' \prec \wv$.\qed
\end{definition}

\sfequilibrium model are similar to epistemic equilibrium model in the sense of~\cite{CerroHS15}, being the major difference that, in our approach, different evaluations in the here world can be used to minimise the same evaluation in the there world.
As a result, we conjeture that every \sfequilibrium model is also an epistemic equilibrium model in the sense of~\cite{CerroHS15}.
On the other hand, the converse does not hold in general.
For instance, if we take the program~\program\ref{theory:larger.set.of.worlds} from Example~\ref{ex:unfounded}, we have that $\sset{\set{a},\set{b}}$ is an epistemic equilibrium model in the sense of~\cite{CerroHS15}, but not an \sfequilibrium model in our approach.
This shows that epistemic equilibrium model in the sense of~\cite{CerroHS15} are not unfounded-free (in the same way that world views or autoepistemic equilibrium model in this approach are not unfounded-free).
The following result shows that
\sfequilibrium model are, in fact,
unfounded-free.

\begin{theorem}\label{thm:unfounded-free.sf}
Any \sfequilibrium model of any program~$\Pi$ is \mbox{unfounded-free.\qed}
\end{theorem}

\begin{proof}
Let $\wv$ be some \sfequilibrium model of some program~$\Pi$
and suppose, for the sake of contradiction, that it is not unfounded-free.
Then, there is a unfounded-set~$\us$ for~$\Pi$ with respect to~$\wv$ such that, for every $\tuple{X,I} \in \us$,
we have  $I \in \wv$ and $X \cap I \neq \emptyset$.
Let
\begin{IEEEeqnarray*}{l ?C? C l C l}
\wv' &=&& \setm{ \tuple{I,I} &}{& \tuple{I,I} \in \wv \text{ and no } \tuple{X,I} \notin \us}
\\
	 &&\cup& \setm{ \tuple{I\setminus X,I} &}{& \tuple{I,I} \in\wv  \text{ and } \tuple{X,I} \in \us}
\end{IEEEeqnarray*}
Since $\us$ is non-empty and $X \cap I \neq\emptyset$ for all $\tuple{X,I} \in \us$, we have that
$\wv' \prec \wv$ and, since $\wv$ is an \sfequilibrium model of~$\Pi$, it must be that $\wv'$ is not an \sfhtmodel of $\Pi$.
Hence, there is some rule \mbox{$r \in \Pi$} such that
$\wv$ is an \sfhtmodel of $r$ while $\wv'$ is not.
Besides, the latter implies that there is some $\tuple{H,T} \in \wv'$
such that $\tuple{\wv',H,T}\not\models r$
and, thus, that one of following conditions must hold:
\begin{enumerate}
\item $\tuple{\wv',H,T} \models \bigwedge\Body(r)$ and $\tuple{\wv',H,T} \not\models \bigvee\Head(r)$, or
	\label{item:1:thm:unfounded-free.sf}
\item $\tuple{\wv,T,T} \models \bigwedge\Body(r)$ and $\tuple{\wv,T,T} \not\models \bigvee\Head(r)$.
\end{enumerate}
Note that the latter is a contradiction with the fact that
$\wv$ is an \sfhtmodel of $r$ and, thus, the former must hold.
Furthermore, $\tuple{\wv',H,T} \models \bigwedge\Body(r)$ implies $\tuple{\wv,T} \models \bigwedge\Body(r)$ and,
since $\wv$ is an \sfhtmodel of $r$,
that $\tuple{\wv,T} \models \bigvee\Head(r)$.
Hence,
$\Head(r) \cap H = \emptyset$
and
there is an atom $a \in \Head(r)$ such that 
$a \in T \setminus H$.
By construction, this implies that, since $\tuple{X,T} \in \us$ with $X = T \setminus H$
and, thus, one of the following conditions must hold:
\begin{enumerate}
\item $\kdint{T}{W} \not\models \bigwedge\Body(r)$, 
\item $\Bodyrp(r) \cap X \neq \emptyset$, or
\item $(\Head(r) \setminus X) \cap T \neq \emptyset$, or
\item $\Bodymp(r) \cap Y \neq \emptyset$.
\end{enumerate}
The first condition cannot hold because we have that $\tuple{\wv,T} \models \bigwedge\Body(r)$.
Furthermore, $X = T \setminus H$,
we have that
\\
$(\Head(r) \setminus X) \cap T \neq \emptyset$
holds
\\iff
$(\Head(r) \setminus (T \setminus H)) \cap T \neq \emptyset$
\\iff
$(\Head(r) \cap \overline{(T \cap \overline{H}))} \cap T \neq \emptyset$
\\iff
$(\Head(r) \cap (\overline{T} \cup H)) \cap T \neq \emptyset$
\\iff
$\Head(r) \cap \overline{T} \cap T \neq \emptyset$ or $\Head(r) \cap H \cap T \neq \emptyset$
\\iff
$\Head(r) \cap H \neq \emptyset$
\\
which does not hold.
Hence, the forth condition cannot hold either.

Assume now that $b \in \Bodyrp(r) \cap X \neq \emptyset$.
Then, $\tuple{\wv',H,T} \not\models b$ and, thus, we have that $\tuple{\wv',H,T}  \not\models \bigwedge\Body(r)$ which is a contradiction
with~\ref{item:1:thm:unfounded-free.sf}.
Therefore, it must be that
$\Bodymp(r) \cap Y \neq \emptyset$ holds.
Pick some atom \mbox{$b \in \Bodymp(r) \cap Y$}.
But then, there is some $\tuple{X',T'} \in \us$ such that $b \in X'$
and $\tuple{H',T'} \in \wv$ with $X' = T' \setminus H'$.
This implies that
$\tuple{\wv',H',T'} \not\models b$ and, thus, that $\tuple{\wv',H,T} \not\models \bL b$ which is a contradiction with the fact that $\tuple{\wv',H,T} \models \bigwedge\Body(r)$.
Consequently, $\wv$ must be unfounded-free.\qed
\end{proof}

\begin{observation}\label{obs:ht.d45models.are.s5models}
Let $\Gamma$ be a theory and~$\cI = \tuple{\wv,H,T}$ be a model of $\Gamma$.
Then, $\wv$ is an \sfhtmodel of~$\Gamma$.\qed
\end{observation}

\begin{proposition}\label{prop:epistemic-eq.are.s5equilibrium}
Let $\Gamma$ be a theory and
$\wv$ be an \sfinterpretation such that $\tuple{\wv,T}$ is a \kdequilibrium model of $\Gamma$
for every $T \in \wv$.
Then, $\wv$ is an \sfequilibrium model of~$\Gamma$.
\end{proposition}

\begin{proof}
First note that, since $\tuple{\wv,T}$ is a \kdequilibrium model of $\Gamma$
for every $T \in \wv$,
it follows that $\tuple{\wv,T}$ is a model of $\Gamma$
for every $T \in \wv$
and, thus, that
$\tuple{\wv,T} \models \varphi$ for every formula~$\varphi \in \Gamma$.
In its turn, this implies that $\wv$ is an \sfmodel of $\Gamma$.
Suppose now, for the sake of contradiction, that $\wv$ is not an \sfequilibrium model of~$\Gamma$
and, thus, that there is some \sfhtmodel~$\wv'$ of $\Gamma$ such that $\wv' \prec \wv$.
Hence, there is $\tuple{H,T} \in \wv'$ such that $H \subset T$.
Furthermore, since  $\tuple{\wv,T}$ is a \kdequilibrium model of $\Gamma$
it follows that $\tuple{\wv',H,T}$ is not a model of $\Gamma$.
Hence, there is a formula~$\varphi \in \Gamma$ and $\tuple{H',T'} \in \wv' \cup \set{\tuple{H,T}} = \wv'$
such that
$\tuple{\wv',H',T'} \not\models \varphi$.
This implies that $\wv'$ is not an \sfhtmodel of $\varphi$
which is a contradiction with the fact that
$\wv'$ is an \sfhtmodel of $\Gamma$.
Consequently, $\wv$ must be an \sfequilibrium model of~$\Gamma$.\qed
\end{proof}

Note that, in general, the converse of Proposition~\ref{prop:epistemic-eq.are.s5equilibrium} does not hold.
For instance, $\sset{\set{a}}$ is the unique \sfequilibrium model of the the theory $\set{\bL a}$
while $\kdint{\set{a}}{\sset{\set{a}}}$ is not a \kdequilibrium model of it.

\begin{corollary}\label{cor:autoepistemic.are.s5equilibrium}
Let $\wv$ be an autoepistemic world view of some theory~$\Gamma$.
Then, it satisfies the following two conditions:
\begin{enumerate}
\item $\wv$ is an \sfequilibrium model, and
\item there is no propositional interpretation $T$ such that model $\tuple{\wv,T}$ is an \kdequilibrium model of $\Gamma$ and $T \notin \wv$.\qed
\end{enumerate}
\end{corollary}

Note that, in general, the converse of Corollary~\ref{cor:autoepistemic.are.s5equilibrium} does not hold.
For instance, theory $\set{\bL a,\neg\neg a}$
has no autoepistemic world view
while we have that
$\sset{\set{a}}$ is an \sfequilibrium model
and
$\kdint{\emptyset}{\sset{\set{a}}}$ is not a \kdequilibrium model of it.
To see that
$\sset{\set{a}}$ is not autoepistemic world view,
note that
$\kdint{\set{a}}{\sset{\set{a}}}$ is not a \kdequilibrium model of $\set{\bL a,\neg\neg a}$.

\begin{Proofof}{\ref{thm:unfounded-free}}
From Corollary~\ref{cor:autoepistemic.are.s5equilibrium},
we have that $\wv$ is an \sfequilibrium model and, from Theorem~\ref{thm:unfounded-free.sf}, this implies that $\wv$ is unfounded-free.
\end{Proofof}

\section{Proof of Theorem~\ref{thm:g91}}

The proof of Theorem~\ref{thm:g91} rely on the definition of weak autoepistemic world views.
We first show that every autoepistemic world view is also a weak autoepistemic world view
and then that
weak autoepistemic world views
coincide with
G91-world views.
Then, we obtain that autoepistemic world views are G91-world views as a corollary.
Let us start by defining semi-total interpretations.
We say that an interpretation $\tuple{\wv,H,T}$ is \emph{semi-total} iff 
every \htinterpretation~$\tuple{H',T'} \in \wv$ is total, that is, $H' = T'$.
It is easy to see that, every total interpretation is semi-total but not vice-versa.

\begin{definition}
A total interpretation $\cI$ is said to be a \emph{weak \kdequilibrium model} of some theory~$\Gamma$ iff $\cI$ is a model of $\Gamma$ and there is no other semi-total model $\cI'$ of $\Gamma$ such that $\cI' \prec \cI$.\qed
\end{definition}

Note that every \kdequilibrium model is also a weak \kdequilibrium model, but not vice-versa.
For instance, $\kdint{\set{a}}{\sset{\set{a}}}$ is a weak \kdequilibrium model of $\set{a \leftarrow \bL a}$
but not a \kdequilibrium model.

\begin{definition}\label{def:weal-au-eqmodel}
A \sfinterpretation~$\wv$ is called a \emph{weak autoepistemic world view} of~$\Gamma$
iff it satisfies the following two conditions:
\begin{enumerate}
\item $\kdint{T}{\wv}$ is a weak \kdequilibrium model of $\Gamma$ for every 
\mbox{$T \in \wv$}, and
\item there not exists any propositional interpretation $T$ such that $\tuple{\wv,T}$ is a weak \kdequilibrium model of $\Gamma$ and $T \notin \wv$.\qed
\end{enumerate}
\end{definition}

\begin{proposition}\label{prop:a.world.view->weak}
Every autoepistemic world view is also a weak autoepistemic world view.\qed
\end{proposition}

\begin{proof}
Since every \kdequilibrium model is also a weak \kdequilibrium model,
it only remains to be shown that if $\wv$ is an autoepistemic world view then
\begin{enumerate}
\item[]there not exists any propositional interpretation $T$ such that $\tuple{\wv,T}$ is a weak \kdequilibrium model of $\Gamma$ and $T \notin \wv$.
\end{enumerate}
Suppose, for the sake of contradiction, that 
there is some propositional interpretation $T$ such that $\tuple{\wv,T}$ is a weak \kdequilibrium model of $\Gamma$ and $T \notin \wv$.
Since $\wv$ is an autoepistemic world view, this implies that
$\tuple{\wv,T}$ is not a \kdequilibrium model of $\Gamma$ and, thus, that there is some non-semi-total model $\cI'=\tuple{\wv',H,T}$ of $\Gamma$ such that $\cI' \prec \cI$.
Hence, there is some $\tuple{H',T'} \in \wv'$ such that $H' \subset T'$.
Let $\cI'' = \tuple{\wv',H',T'}$.
Then, we have that $\cI'' \prec \tuple{\wv,T'}$ and, since $\wv$ is an autoepistemic world view of $\Gamma$ and $T \in \wv$, we have that $\tuple{\wv,T'}$ is a \kdequilibrium model of $\Gamma$.
These two facts together imply that $\cI''$ is not a model of $\Gamma$.
Hence, there is a formula~$\varphi\in\Gamma$ such that $\cI''$ is not a model of $\varphi$ and, thus,
there is $\tuple{H'',T''} \in \wv' \cup \set{\tuple{H',T'}}$ such that
$\tuple{\wv',H'',T''} \not\models \varphi$.
On the other hand, since
$\cI'$ is a model of $\Gamma$ ,
it follows that
$\tuple{\wv',H''',T'''} \not\models \varphi$
for every $\tuple{H''',T'''} \in \wv' \cup \set{\tuple{H,T}}$.
Hence, it follows that $H''=H'$ and $T''=T'$ and that
$\tuple{\wv',H',T'} \not\models \varphi$.
However, since $\tuple{H',T'} \in \wv'$,
this implies that
$\cI' = \tuple{\wv',H,T} \not\models \varphi$, which is a contradiction with the fact that $\cI'$ is a model of $\Gamma$.
Consequently,
$\wv$ is a weak autoepistemic world view.\qed
\end{proof}

\begin{lemma}\label{lem:semi-tota.reduct}
Let $\Gamma$ be a formula and $\cI = \tuple{\wv,H,T}$ be a semi-total interpretation.
Then, $\cI$ is a model of $\varphi$ iff $\cI$ is a model of $\varphi^\wv$.\qed
\end{lemma}

\begin{proof}
Assume that $\varphi = \bL \psi$.
Then, we have that
$\tuple{\wv,H',T'} \models \varphi$ if and only if
$\tuple{\wv,T'',T''} \models \psi$ for every $\tuple{T'',T''} \in \wv$
iff
$\wv$ is a \sfmodel of $\varphi$
iff $\varphi^\wv = \top$
iff $\cI \models \varphi^\wv$.
Then, by induction in the structure of~$\varphi$,
we get that
$\cI \models \varphi$
iff
$\cI \models \varphi^\wv$.
Finally, we have that $\cI$ is a model of $\varphi$
iff $\cI \models \varphi$ and $\tuple{\wv,T',T'} \models \varphi$ for every $\tuple{T',T'} \in \wv$
iff $\cI \models \varphi^\wv$ and $\tuple{\wv,T',T'} \models \varphi^\wv$ for every $\tuple{T',T'} \in \wv$
iff  $\cI$ is a model of $\varphi^\wv$.
\end{proof}

\begin{lemma}\label{lem:ht.correspondence}
Let $\Gamma$ be a propositional theory and $\cI = \tuple{\wv,H,T}$ be some interpretation.
Then, $\cI$ is a model of $\Gamma$ iff $\tuple{H',T'}$ is a
\htmodel of $\Gamma$ 
for every \htinterpretation $\tuple{H',T'} \in \wv \cup \set{ \tuple{H,T} }$.\qed
\end{lemma}

\begin{proof}
By definition, we have that
$\cI$ is a model of $\Gamma$
iff
$\cI$ is a model of $\varphi$
for all $\varphi \in \Gamma$.
Furthermore,
$\cI$ is a model of $\varphi$
iff $\cI \models \varphi$
and
$\tuple{\wv,H',T'} \models \varphi$ for every $\tuple{H',T'} \in \wv$
iff
$\tuple{\wv,H',T'} \models \varphi$ for every \htinterpretation $\tuple{H',T'} \in \wv \cup \set{ \tuple{H,T} }$.
Finally, since $\varphi$ is a propositional formula, we have that
$\tuple{\wv,H',T'} \models \varphi$
iff
$\tuple{H',T'} \models \varphi$.
Hence,
$\cI$ is a model of $\Gamma$ iff $\tuple{H',T'}$ is a
\htmodel of $\Gamma$ 
for every \htinterpretation $\tuple{H',T'} \in \wv \cup \set{ \tuple{H,T} }$.\qed
\end{proof}

\begin{lemma}\label{lem:ht.correspondence}
Let $\Gamma$ be a propositional theory and $\cI = \tuple{\wv,T}$ be some total interpretation.
Then, $\cI$ is a \kdequilibrium model of $\Gamma$ iff $T'$ is an equilibrium model of~$\Gamma$ for every $T' \in \wv \cup \set{T}$.\qed
\end{lemma}

\begin{proof}
First note that, since $\Gamma$ is propositional,
$\cI$ is a model of $\Gamma$
iff
$T'$ is an model of~$\Gamma$ for every $T' \in \wv \cup \set{T}$.
Hence, we have that $\cI$ is a \kdequilibrium model of $\Gamma$ iff
there is not model $\cI'=\tuple{\wv',H,T}$ of $\Gamma$ such that $\cI' \prec \cI$.

Suppose, for the sake of contradiction, that there is some $T' \in \wv \cup \set{T}$
which is not an equilibrium model of $\Gamma$.
Then, there is $\tuple{H',T'} \models \Gamma$ such that $H' \subset T'$.
Let $\cI'=\tuple{\wv,H',T'}$ if $T' = T$ and $\cI'=\tuple{\wv',T',T'}$ with $\wv' = \set{\tuple{H',T'}}\cup \wv$ otherwise.
Then, $\cI' \prec \cI$ and, from Lemma~\ref{lem:ht.correspondence}, we have that $\cI'$ is a model of $\Gamma$ which is a contradiction.
Hence, $T'$ must be an equilibrium model of~$\Gamma$ for every $T' \in \wv \cup \set{T}$.
The other way around.
Assume that
$T'$ is an equilibrium model of~$\Gamma$ for every $T' \in \wv \cup \set{T}$
and suppose, for the sake of contradiction, that $\cI$ is not a \kdequilibrium model of $\Gamma$.
Then, there is some model  $\cI'=\tuple{\wv',H,T}$ of $\Gamma$ such that $\cI' \prec \cI$.
Then, from Lemma~\ref{lem:ht.correspondence},
we have that $\tuple{H',T'} \models\Gamma$ for every
$\tuple{H',T'} \in \wv \cup \set{ \tuple{H,T} }$.
This implies that $T'$ is not an equilibrium model of~$\Gamma$
which is a contradiction.
Consequently,
$\cI$ must be a \kdequilibrium model of~$\Gamma$.\qed
\end{proof}

\begin{theorem}\label{thm:g91.weak}
Given a theory~$\Gamma$ and some \sfinterpretation~$\wv$,
we have that $\wv$ is a weak autoepistemic world view of~$\Gamma$ iff $\wv$ is a G91-world view of $\Gamma$.\qed
\end{theorem}

\begin{proof}
It is easy to see Definition~\ref{def:weal-au-eqmodel} can be rewritten as the following fixpoint equation:
\begin{gather*}
\wv \ \ = \ \ \setm{ T }{ \kdint{T}{\wv } \in \WEQ[ \Gamma ] }
\end{gather*}
where $\WEQ[ \Gamma ]$ denotes the set of all weak \kdequilibrium models of~$\Gamma$.
Furthermore, from Lemmas~\ref{lem:semi-tota.reduct} and~\ref{lem:ht.correspondence},
we have that
$\tuple{\wv,T} \in \WEQ[ \Gamma ]$ iff $\tuple{\wv,T} \in \WEQ[ \Gamma^\wv ]$
iff $\wv \cup \set{T} \subseteq \EQ[ \Gamma^\wv ]$
and, therefore, the above can be rewriting as
\begin{gather*}
\wv \ \ = \ \ \setm{ T }{ \wv \cup \set{T} \subseteq \EQ[ \Gamma^\wv ] }
\end{gather*}
which holds iff $\wv = \EQ[ \Gamma^\wv ]$.
That is, iff $\wv$ is a G91-world view of $\Gamma$.\qed
\end{proof}

Note that weak autoepistemic world views are very similar to epistemic views as defined in~\cite{WangZ05} and, in fact, we conjecture that they coincide for any theory without nested implications.
On the other hand, the following example shows that this does not hold in general:

\begin{example}
Consider the singleton theory $\Gamma = \set{ \neg\neg a \wedge \bL\varphi \to a }$
with $\varphi$ the following formula $\varphi = \neg\neg a \to a$.
Then, $\wv = \sset{\emptyset,\set{a}}$ is both a G91 and autoepistemic world view of~$\Gamma$, but not an epistemic view.
To see that $\wv$ is both a G91-world view of~$\Gamma$,
note that $\Gamma^\wv = \set{ \neg\neg a \wedge \top \to a } \equiv \set{ \neg\neg a \to a }$
which has two stable models: $\emptyset$ and $\set{a}$.
Hence, $\wv$ is a G91-world view of~$\Gamma$.
Furthermore, from Theorem~\ref{thm:g91.weak} this implies that $\wv$ is also an autoepistemic world view.
On the other hand, $\wv$ is not an epistemic view
because 
$\tuple{\wv,\emptyset,\set{a}}$ is a model of $\Gamma$ in the sense of~\cite{WangZ05}.
Note that
$\tuple{\wv,\emptyset,\set{a}} \not\models \varphi$ and that both $\emptyset$ and $\set{a}$ belong to $\wv$.
This implies that
$\tuple{\wv,\emptyset,\set{a}} \not\models \bL\varphi$
and, thus, that $\tuple{\wv,\emptyset,\set{a}}$ is a model of $\Gamma$.
On the other hand,
in our logic
$\tuple{\wv,\emptyset,\set{a}} \not\models \varphi$ 
does not imply
$\tuple{\wv,\emptyset,\set{a}} \not\models \bL\varphi$.
In fact
$\tuple{\wv,\emptyset,\set{a}} \models \bL\varphi$ holds
because both
$\tuple{\wv,\emptyset,\emptyset} \models \varphi$ 
and
$\tuple{\wv,\set{a},\set{a}} \models \varphi$ 
hold.
\qed
\end{example}

This example shows that epistemic views~\cite{WangZ05} are different from G91-world views as defined in~\cite{Truszczynski11} and from ours weak autoepistemic world views.

\begin{Proofof}{\ref{thm:g91}}
From Proposition~\ref{prop:a.world.view->weak},
we have that every autoepistemic world view is also a weak autoepistemic world view and, from Theorem~\ref{thm:g91.weak},
this implies that every autoepistemic world view is also a G91-world view.
\end{Proofof}

\section{Proof of the Main Theorem}

First note that, by taking together Theorems~\ref{thm:unfounded-free} and~\ref{thm:g91},
we immediately obtain the following result:

\begin{corollary}\label{cor:awv->g91wv}
Any autoepistemic world view~$\wv$ of any program~$\Pi$ is also an unfounded-free G91-world view.\qed
\end{corollary}

The converse of Corollary~\ref{cor:awv->g91wv} is proved in two steps.
We first prove that every unfounded-free \sfmodel is also an \sfequilibrium model and then that every \sfequilibrium model which is also a weak autoepistemic world view is an autoepistemic world view.
As a Corollary, we obtain that every unfounded-free G91-world view is also an
autoepistemic world view.

\begin{theorem}\label{thm:unfounded-free.sf.iff}
Given any program~$\Pi$ and some \sfinterpretation~$\wv$,
we have that $\wv$ is an \sfequilibrium model iff $\wv$ is an unfounded-free \sfmodel.\qed
\end{theorem}

\begin{proof}
First note that from Theorem~\ref{thm:unfounded-free.sf},
we have that any
\sfequilibrium model of any program~$\Pi$ is unfounded-free.
Suppose now, for the sake of contradiction, that $\wv$ is an \sfmodel of~$\Pi$ but not an \sfequilibrium model.
Then, there is some \sfhtmodel~$\wv'$ of $\Pi$ such that $\wv' \prec \wv$.
Let us define
\begin{gather*}
\us \ \ \eqdef \ \ 
	\setbm{ \tuple{X,I} }{ \tuple{H,I} \in \wv' \text{ with } X = I \setminus H \text{ and } X \neq \emptyset }
\end{gather*}
Then, is clear that, for every $\tuple{X,I} \in \us$ we have $I \in \wv'$ and $X \cap I \neq \emptyset$.
Hence, since $\wv$ is unfounded-free, it follows that $\us$ cannot be an unfounded-witness for $\Pi$ w.r.t. $\wv'$.
Besides, since $\wv' \prec \wv$, we have that $\us \neq \emptyset$ and, thus, there must be some $\tuple{X,I} \in \us$, atom $a \in X$ and rule $r \in \Pi$ with $a \in \Head(r)$ satisfying all the following conditions:
\begin{enumerate}
\item $\kdint{I}{\wv} \models \bigwedge\Body(r)$, and
	\label{item:1:prop:s5-unfounded-free}
\item $\Bodyrp(r) \cap X = \emptyset$, and
	\label{item:2:prop:s5-unfounded-free}
\item $(\Head(r) \setminus X) \cap I = \emptyset$.
	\label{item:4:prop:s5-unfounded-free}
\item $\Bodymp(r) \cap Y = \emptyset$, and
	\label{item:3:prop:s5-unfounded-free}
\end{enumerate}
Furthermore, since $\wv'$ is an \sfhtmodel of $\Pi$, we have that
$\kdint{H,I}{\wv'} \models r$ for every $r \in \Pi$ and
with $H = I \setminus X$
and that
one of the following must hold:
\begin{enumerate}[ start=5]
\item $\kdint{I}{\wv} \not\models \bigwedge\Body(r)$, or
	\label{item:5:prop:s5-unfounded-free}
\item $\kdint{H,I}{\wv'} \models \bigvee\Head(r)$, or
	\label{item:6:prop:s5-unfounded-free}
\item $\kdint{H,I}{\wv'} \not\models \bigwedge\Body(r)$ and $\kdint{I}{\wv} \models \bigvee\Head(r)$.
	\label{item:7:prop:s5-unfounded-free}
\end{enumerate}
Clearly, \eqref{item:5:prop:s5-unfounded-free} is in contradiction with \eqref{item:1:prop:s5-unfounded-free}
and, thus, either \eqref{item:6:prop:s5-unfounded-free} or \eqref{item:7:prop:s5-unfounded-free} must hold.
Note also that~\eqref{item:4:prop:s5-unfounded-free} implies that
$\Head(r) \cap H = \emptyset$
and, thus, \eqref{item:6:prop:s5-unfounded-free} cannot hold either.
In more detail, we have that
\\$(\Head(r) \setminus X) \cap I = \emptyset$ holds
\\iff
$(\Head(r) \setminus(I \setminus H)) \cap I = \emptyset$
\\iff
$(\Head(r) \cap \overline{ (I \cap \overline{H}))} \cap I = \emptyset$
\\iff
$(\Head(r) \cap (\overline{I} \cup H)) \cap I = \emptyset$
\\iff
$\Head(r) \cap\overline{I} \cap I = \emptyset$ and 
$\Head(r) \cap H \cap I = \emptyset$
\\iff
$\Head(r) \cap H \cap I = \emptyset$
\\iff
$\Head(r) \cap H = \emptyset$
\\iff
$\kdint{H,I}{\wv'} \not\models \bigvee\Head(r)$
Hence, \eqref{item:7:prop:s5-unfounded-free} must hold and, thus, there is some literal $L \in \Bodyp(r)$
such that
$\kdint{H,I}{\wv'} \not\models L$.
If $L \in \Bodyr(r)$, then, we have that
$L \notin H$.
Besides, from condition~\ref{item:1:prop:s5-unfounded-free}, we have that $L \in I$
and, thus, $L \in I \setminus H = X$ which is a contradiction
with condition~\ref{item:2:prop:s5-unfounded-free}.
Otherwise,
$L \in \Bodym(r)$ and we have that there is $\tuple{H',I'} \in \wv'$
such that $a \notin H'$ with $L = \bL a$.
But then, $\tuple{X',I'} \in \us$ with $X' = I' \setminus H'$ and, thus, $a \in Y$ which is a contradiction
with condition~\ref{item:3:prop:s5-unfounded-free}.
Consequently,
$\wv$ is an \sfequilibrium model.\qed
\end{proof}

\begin{lemma}\label{lem:epistemic.decomposition}
Let $\Gamma$ be a theory,
$\wv$ be an \sfequilibrium model of~$\Gamma$ and $T \in \wv$ be a propositional interpretation.
If $\tuple{\wv,T}$ is a weak \kdequilibrium model of~$\Gamma$,
then $\tuple{\wv,T}$ is also a \kdequilibrium model of~$\Gamma$.
\end{lemma}

\begin{proof}
Suppose that $\tuple{\wv,T}$ is not a \kdequilibrium model of $\Gamma$.
Since $\tuple{\wv,T}$ is a weak \kdequilibrium model of $\Gamma$,
we have that $\kdint{T}{\wv}$ is model of $\Gamma$
and, thus, there must be some non-semi-total model $\cI'=\tuple{\wv',H,T}$ of $\Gamma$ such that
$\cI' \prec \kdint{T}{\wv}$.
Hence, there is $\tuple{H',T'} \in \wv'$ with $H \subset T$.
This implies that $\wv' \prec \wv$ and, since $\wv$ is an \sfequilibrium model of~$\Gamma$,
that $\wv$ is not a \sfhtmodel of~$\Gamma$.
Hence, $\wv'$ is not a \sfhtmodel of $\Gamma$.
This implies that there is some formula $\varphi \in \Gamma$ and \htinterpretation~$\tuple{H'',T''} \in \wv'$
such that $\tuple{\wv',H'',T''} \not\models \varphi$.
In its turn, this implies that $\cI'$ is not a model of $\Gamma$ which is a contradiction.
Consequently, $\tuple{\wv,T}$ is a \kdequilibrium model of~$\Gamma$.\qed
\end{proof}

\begin{proposition}\label{prop:auto-wv<->sf-eq+weak-auto-wv}
Given a theory~$\Gamma$ and an \sfinterpretation~$\wv$,
we have that $\wv$ is an autoepistemic world view iff $\wv$ is both an \sfequilibrium model and a weak autoepistemic world view of~$\Gamma$.\qed
\end{proposition}

\begin{proof}
First note that from Corollary~\ref{cor:autoepistemic.are.s5equilibrium} and Proposition~\ref{prop:a.world.view->weak},
we have that every autoepistemic world view is both an \sfequilibrium model and a weak autoepistemic world view of~$\Gamma$.
The other way around.
Since $\wv$ is a weak autoepistemic equilibrium model of $\Gamma$,
we have that $\tuple{\wv,T}$ is a weak \kdequilibrium model of $\Gamma$ for every $T \in \wv$
which, from Lemma~\ref{lem:epistemic:prop:decomposition},
implies that
$\tuple{\wv,T}$ is a \kdequilibrium model of $\Gamma$ for every $T \in \wv$.
Hence, it only remains to be shown that the following condition hold:
\begin{enumerate}
\item[] there is no propositional interpretation $T$ such that model $\tuple{\wv,T}$ is an \kdequilibrium model of $\Gamma$ and $T \notin \wv$.
\end{enumerate}
Suppose, for the sake of contradiction, that
there is some propositional interpretation $T$ such that model $\tuple{\wv,T}$ is an \kdequilibrium model of $\Gamma$ and $T \notin \wv$.
First, this implies that $\tuple{\wv,T}$ is an \kdmodel of $\Gamma$.
Then, since $\wv$ is a weak autoepistemic equilibrium model of $\Gamma$ and $T \notin \wv$,
it follows that 
$\tuple{\wv,T}$
cannot be a weak \kdequilibrium model of $\Gamma$ and, thus,
there is some semi-total model $\cJ$ is a model of $\Gamma$ such that $\cJ \prec \tuple{\wv,T}$.
But every semi-total model is also a model, so this is a contradiction with the fact that $\tuple{\wv,T}$ is a \kdequilibrium model of $\Gamma$.
Consequently, $\wv$ must be an autoepistemic world view of $\Gamma$.\qed
\end{proof}

\begin{corollary}\label{cor:auto-wv<->unfounded-free.weak}
Given any program~$\Pi$ and some \sfinterpretation~$\wv$,
we have that $\wv$ is an autoepistemic world view of $\Pi$ iff $\wv$ is an unfounded-free weak autoepistemic world view.\qed
\end{corollary}

\begin{proof}
From Proposition~\ref{prop:auto-wv<->sf-eq+weak-auto-wv} and Theorem~\ref{thm:unfounded-free.sf.iff}, we have that
$\wv$ is an autoepistemic world view iff $\wv$ is both an \sfequilibrium model and a weak autoepistemic world view iff $\wv$ is an unfounded-free \sfmodel and a weak autoepistemic world view.
Finally, it is easy to see that every weak autoepistemic world view is also an \sfmodel and, thus, we get that
any $\sfinterpretation$ is an autoepistemic world view iff  is an unfounded-free weak autoepistemic world view.\qed
\end{proof}

\begin{proofof}{the Main Theorem}
Let~$\wv$ be any \sfinterpretation.
Then, from Corollary~\ref{cor:auto-wv<->unfounded-free.weak},
we have that $\wv$ is an autoepistemic world view iff $\wv$ is an unfounded-free weak autoepistemic world view.
Furthermore, from Theorem~\ref{thm:g91.weak},
we have that $\wv$ is a  weak autoepistemic world view iff $\wv$ is a G91-world view.
Taking these two facts together, we get that
$\wv$ is an autoepistemic world view iff $\wv$ is an unfounded-free G91-world view.
\end{proofof}

\section{Proofs of Propositions~\ref{prop:neg-L} and~\ref{prop:properties}}

\begin{proposition}\label{prop:negation.there}
Let $\varphi$ be a formula 
and let $\tuple{\wv,H,T}$ be an interpretation.
Then, $\tuple{\wv,H,T}\models \neg\varphi$ iff $\tuple{\wv^t,T,T} \not\models \varphi$ iff $\tuple{\wv^t,T,T} \models \neg\varphi$.\qed
\end{proposition}

\begin{proof}
By definition,
we have that
$\tuple{\wv^t,T,T} \models \neg\varphi$
iff
$\tuple{\wv^t,T,T} \models \varphi\to\bot$
iff
$\tuple{\wv^t,T,T} \not\models \varphi$.
Furthermore, by definition, we have that
$\tuple{\wv,H,T}\models \neg\varphi$
iff both
\mbox{$\tuple{\wv,H,T} \not\models \varphi$}
and
$\tuple{\wv,T,T} \not\models \varphi$.
Finally, since Proposition~\ref{prop:pers},
we have that
$\tuple{\wv,H,T}\models \varphi$
implies
$\tuple{\wv,T,T}\models\varphi$
we get that
$\tuple{\wv,H,T}\models \neg\varphi$
iff
$\tuple{\wv,T,T} \not\models \varphi$.\qed
\end{proof}

\begin{lemma}\label{lem:satisfication:prop:neg-L}
Let $\varphi$ be a formula where all occurrences of $\bL$ are in the scope of negation
and let $\tuple{\wv,H,T}$ be an interpretation.
Then, $\tuple{\wv,H,T}\models \varphi$ iff $\tuple{\wv^t,H,T}\models \varphi$.\qed
\end{lemma}

\begin{proof}
In case that $\varphi \in \at$ is an atom, we have that
$\tuple{\wv,H,T}\models \varphi$ iff $\varphi \in H$ iff
$\tuple{\wv^t,H,T}\models \varphi$.
In case $\varphi = \neg\psi$, 
from Proposition~\ref{prop:negation.there},
we have that
$\tuple{\wv,H,T}\models \varphi$
iff $\tuple{\wv^t,T,T}\models \varphi$
iff $\tuple{\wv^t,H,T}\models \varphi$.
Otherwise, the proof follows by induction in the structure of the formula.\qed
\end{proof}

\begin{lemma}\label{lem:wv:prop:neg-L}
Let $\Gamma$ be a formula where all occurrences of $\bL$ are in the scope of negation
and let $\wv$ be a weak autoepistemic world view.
Then, $\wv$ is also an autoepistemic world view.\qed
\end{lemma}

\begin{proof}
Since $\wv$ be a weak autoepistemic world view we have that
\begin{enumerate}
\item $\kdint{T}{\wv}$ is a weak \kdequilibrium model of $\Gamma$ for every 
\mbox{$T \in \wv$}, and
	\label{itme:1:lem:models:prop:neg-L}

\item there not exists any propositional interpretation $T$ such that $\tuple{\wv,T}$ is a weak \kdequilibrium model of $\Gamma$ and $T \notin \wv$.
	\label{itme:2:lem:models:prop:neg-L}
\end{enumerate}
First note that condition~\ref{itme:2:lem:models:prop:neg-L}
immediately implies that  there not exists any propositional interpretation $T$ such that $\tuple{\wv,T}$ is a \kdequilibrium model of $\Gamma$ and $T \notin \wv$.
Suppose now, for the sake of contradiction, that
\begin{enumerate}
\item[] there is \mbox{$T \in \wv$} such that $\kdint{T}{\wv}$ is a not \kdequilibrium model of $\Gamma$.
\end{enumerate}
Hence, there is a non-semi-total model $\cI = \kdint{H,T}{\wv'}$ of $\Gamma$ s.t. $\cI \prec \kdint{T}{\wv}$.
Furthermore, since $\cI$ is a model of $\Gamma$,
we have that
$\kdint{H',T'}{\wv'} \models \varphi$  all $\varphi \in \Gamma$ and all $\tuple{H',T'} \in \wv \cup \set{\tuple{H,T}}$
which, from Lemma~\ref{lem:satisfication:prop:neg-L} and the fact that all occurrences of $\bL$ in $\varphi$ are in the scope of negation,
implies that
$\kdint{H',T'}{\wv} \models \varphi$ holds for all $\varphi \in \Gamma$ and all $\tuple{H',T'} \in \wv \cup \set{\tuple{H,T}}$.
Note that, since $\cI \prec \kdint{T}{\wv}$,
there is $\tuple{H'',T''} \in \wv' \cup \set{\tuple{H,T}}$ with $H'' \subset T''$
and we have that $\kdint{H'',T''}{\wv}$ is a semi-total model of $\Gamma$ and that $\kdint{H'',T''}{\wv} \prec \kdint{T''}{\wv}$.
Note also that $\tuple{H'',T''} \in \wv'$ implies $T'' \in \wv$ and, thus, the above is a contradiction with condition~\ref{itme:1:lem:models:prop:neg-L}.
Hence, we have that
 $\kdint{T}{\wv}$ is a \kdequilibrium model of $\Gamma$ for every \mbox{$T \in \wv$}.
Consequently, $\wv$ is an autoepistemic world view.\qed
\end{proof}

\begin{Proofof}{\ref{prop:neg-L}}
From Proposition~\ref{prop:a.world.view->weak},
we have that every autoepistemic world view~$\wv$ of any theory~$\Gamma$ is also a weak autoepistemic world view.
Furthermore, since all occurrences of $\bL$ in $\Gamma$ are in the scope of negation,
from Lemma~\ref{lem:wv:prop:neg-L}, the converse also holds.
\end{Proofof}

\begin{lemma}\label{lem:subjective.formulas}
Let $\varphi$ be a formula in which every atom is in the scope of the modal operator~$\bL$
and $\cI=\tuple{\wv,H,T}$ be some interpretation.
Then, $\cI \models \varphi$ iff $\wv^t$ is an \sfmodel of $\varphi$.
\end{lemma}

\begin{proof}
In case $\varphi = \bL \psi$,
we have that
$\cI \models \varphi$
iff
$\tuple{\wv,H',T'} \models \psi$ for all \mbox{$\tuple{H',T'} \in \wv$}
iff $\wv$ is an \sfhtmodel of $\varphi$.
The rest of the proof follows by induction on the structure of~$\varphi$.\qed
\end{proof}

\begin{lemma}\label{lem:subjective.constraints}
Let $\varphi$ be a formula in which every atom is in the scope of the modal operator~$\bL$
and $\cI=\tuple{\wv,H,T}$ be some interpretation.
Then, $\cI$ is a model of $\bot\leftarrow\varphi$ iff $\wv^t$ is an \sfmodel of $\bot\leftarrow\varphi$.
\end{lemma}

\begin{proof}
By definition, we have that
$\cI$ is a model of $\bot\leftarrow\varphi$
\\iff
$\tuple{\wv,H',T'} \models \bot\leftarrow\varphi$
for all $\tuple{H',T'} \in \wv \cup\set{\tuple{H,T}}$
\\iff
$\tuple{\wv^t,T',T'} \not\models \varphi$
for all $\tuple{H',T'} \in \wv \cup\set{\tuple{H,T}}$
\\iff
$\wv^t$ is not an \sfhtmodel of $\varphi$ (Lemma~\ref{lem:subjective.formulas})
\\iff
$\wv^t$ is an \sfhtmodel of $\bot\leftarrow\varphi$
\\iff
$\wv^t$ is an \sfmodel of $\bot\leftarrow\varphi$.\qed
\end{proof}

\begin{lemma}\label{lem:theory.with.subjective.constraints}
Let $\Gamma$ be a theory and $\varphi$ be a formula in which every atom is in the scope of the modal operator~$\bL$
and $\cI=\tuple{\wv,H,T}$ be some interpretation.
Then, $\cI$ is a \kdequilibrium model of $\Gamma \cup \set{\bot\leftarrow\varphi}$ iff
$\cI$ is a \kdequilibrium model of $\Gamma$
and
$\wv^t$ is an \sfmodel of $\bot\leftarrow\varphi$.
\end{lemma}

\begin{proof}
Assume first that $\cI$ is a \kdequilibrium model of $\Gamma \cup \set{\bot\leftarrow\varphi}$.
Then, $\cI$ is a model of $\Gamma$ and a model of $\bot\leftarrow\varphi$.
From Lemma~\ref{lem:subjective.constraints}, the latter implies that
$\wv^t$ is an \sfmodel of $\bot\leftarrow\varphi$.
Suppose, for the sake of contradiction, that $\cI$ is a \kdequilibrium model of $\Gamma$ and, thus, that there is some model $\cI' = \tuple{\wv',H',T'}$ of $\Gamma$ such that $\cI' \prec \cI$.
Then, $\cI' \prec \cI$ implies that $(\wv')^t = \wv$ and, thus, 
from Lemma~\ref{lem:subjective.constraints} and the fact that
$\wv^t$ is an \sfmodel of $\bot\leftarrow\varphi$,
it follows that $\cJ$ is a model of $\bot\leftarrow\varphi$
which is a contradiction with the fact that $\cI$ is a \kdequilibrium model of $\Gamma \cup \set{\bot\leftarrow\varphi}$.

The other way around.
Assume that $\cI$ is a \kdequilibrium model of $\Gamma$
and
$\wv^t$ is an \sfmodel of $\bot\leftarrow\varphi$.
From Lemma~\ref{lem:subjective.constraints}, we have that
$\cI$ is a model of $\bot\leftarrow\varphi$
and, thus, a \kdequilibrium model of $\Gamma \cup \set{\bot\leftarrow\varphi}$.\qed
\end{proof}

\begin{Proofof}{\ref{prop:properties}}
For \emph{Property~\ref{property:supraASP}}, note that since we are dealing with propositional theories there is no occurrence of $\bL$ and, thus, we can apply Proposition~\ref{prop:neg-L} and Theorem~\ref{thm:g91.weak} to see that the autoepistemic and G91 world views coincide.
Then, just note that, since there is no occurrence of $\bL$, we have that $\Gamma^\wv = \Gamma$
and, thus $\wv$ is an autoepistemic world view of $\Gamma$ iff $\wv$ is a G91-world view of $\Gamma$ iff $\wv = \SM[\Gamma^\wv] = \SM[\Gamma]$.
\\[10pt]
For \emph{Property~\ref{property:supraS5}}, from Corollary~\ref{cor:autoepistemic.are.s5equilibrium},
it follows that every autoepistemic world view~$\wv$ is also an \sfequilibrium model and it is easy to check that every \sfequilibrium model is also an \sfmodel.
\\[10pt]
For \emph{Property~\ref{property:constraint.monotonicity}},
we have that $\wv$ is an autoepistemic world view of~$\Gamma\cup\set{\bot\leftarrow\varphi}$ iff the following equality holds:
\begin{gather*}
\wv \ \ = \ \ \setm{ T }{ \kdint{T}{\wv } \in \MEQ[ \Gamma \cup\set{\bot\leftarrow\varphi} ] }
\end{gather*}
Furthermore, from Lemma~\ref{lem:theory.with.subjective.constraints} and the fact that 
every atom in $\varphi$ is in the scope of the modal operator~$\bL$,
it follows that
$\kdint{T}{\wv } \in \MEQ[ \Gamma \cup\set{\bot\leftarrow\varphi} ]$
iff
$\kdint{T}{\wv } \in \MEQ[ \Gamma ]$ 
and $\wv$ is \sfmodel of $\bot\leftarrow\varphi$.
Hence, we have that
$\wv$ is an autoepistemic world view of~$\Gamma\cup\set{\bot\leftarrow\varphi}$ iff the following equality holds:
\begin{gather*}
\wv \ \ = \ \ \setm{ T }{ \kdint{T}{\wv } \in \MEQ[ \Gamma ] }
\end{gather*}
and $\wv$ is \sfmodel of $\bot\leftarrow\varphi$.
\end{Proofof}

\end{document}